\documentclass[sigconf]{aamas}

\usepackage{balance} 

\usepackage{amsmath}

\usepackage{amssymb}
\usepackage{mathtools}
\usepackage{amsthm}
\usepackage{algorithm}
\usepackage{algorithmic}
\usepackage{graphicx}
\usepackage{subfigure}
\usepackage{wrapfig}

\theoremstyle{plain}
\newtheorem{theorem}{Theorem}[section]

\newtheorem{lemma}[theorem]{Lemma}

\theoremstyle{definition}

\theoremstyle{remark}

\usepackage{balance}

\makeatletter
\gdef\@copyrightpermission{
  \begin{minipage}{0.2\columnwidth}
   \href{https://creativecommons.org/licenses/by/4.0/}{\includegraphics[width=0.90\textwidth]{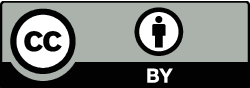}}
  \end{minipage}\hfill
  \begin{minipage}{0.8\columnwidth}
   \href{https://creativecommons.org/licenses/by/4.0/}{This work is licensed under a Creative Commons Attribution International 4.0 License.}
  \end{minipage}
  \vspace{5pt}
}
\makeatother

\setcopyright{ifaamas}
\acmConference[AAMAS '25]{Proc.\@ of the 24th International Conference
on Autonomous Agents and Multiagent Systems (AAMAS 2025)}{May 19 -- 23, 2025}
{Detroit, Michigan, USA}{Y.~Vorobeychik, S.~Das, A.~Nowé  (eds.)}
\copyrightyear{2025}
\acmYear{2025}
\acmDOI{}
\acmPrice{}
\acmISBN{}

\acmSubmissionID{159}

\title[AAMAS-2025 Formatting Instructions]{Dual Ensembled Multiagent Q-Learning with Hypernet Regularizer}

\author{Yaodong Yang}
\affiliation{
  \institution{Department of CSE, CUHK}
  \city{Hong Kong}
  \country{China}}
\email{yydapple@gmail.com}

\author{Guangyong Chen}
\thanks{Corresponding author: Guangyong Chen.}
\affiliation{
  \institution{Zhejiang Lab}
  \city{Hangzhou}
  \country{China}}
\email{gychen@zhejianglab.com}

\author{Hongyao Tang}
\affiliation{
  \institution{Mila, University of Montreal}
  \city{Montreal}
  \country{Canada}}
\email{tang.hongyao@quebec.mila}

\author{Furui Liu}
\affiliation{
  \institution{Zhejiang Lab}
  \city{Hangzhou}
  \country{China}}
\email{liufurui@zhejianglab.com}

\author{Danruo Deng}
\affiliation{
  \institution{Department of CSE, CUHK}
  \city{Hong Kong}
  \country{China}}
\email{drdeng@link.cuhk.edu.hk}

\author{Pheng Ann Heng}
\affiliation{
  \institution{Department of CSE, CUHK}
  \city{Hong Kong}
  \country{China}}
\email{pheng@cse.cuhk.edu.hk}

\begin{abstract}
Overestimation in single-agent reinforcement learning has been extensively studied. In contrast, overestimation in the multiagent setting has received comparatively little attention although it increases with the number of agents and leads to severe learning instability. Previous works concentrate on reducing overestimation in the estimation process of target Q-value. They ignore the follow-up optimization process of online Q-network, thus making it hard to fully address the complex multiagent overestimation problem. To solve this challenge, in this study, we first establish an iterative estimation-optimization analysis framework for multiagent value-mixing Q-learning. Our analysis reveals that multiagent overestimation not only comes from the computation of target Q-value but also accumulates in the online Q-network's optimization. Motivated by it, we propose the Dual Ensembled Multiagent Q-Learning with Hypernet Regularizer algorithm to tackle multiagent overestimation from two aspects. First, we extend the random ensemble technique into the estimation of target individual and global Q-values to derive a lower update target. Second, we propose a novel hypernet regularizer on hypernetwork weights and biases to constrain the optimization of online global Q-network to prevent overestimation accumulation. Extensive experiments in MPE and SMAC show that the proposed method successfully addresses overestimation across various tasks\footnote{Code is available at \href{https://github.com/CNDOTA/AAMAS25-DEMAR}{GitHub}.}.
\end{abstract}


\newcommand{\BibTeX}{\rm B\kern-.05em{\sc i\kern-.025em b}\kern-.08em\TeX}

\begin{document}


\pagestyle{fancy}
\fancyhead{}


\maketitle 


\section{Introduction}
\label{sec:introduction}
Overestimation is a critical challenge for reinforcement learning that stems from the maximum operation when bootstrapping the target Q-value \cite{Thrun-1993-issues}. This Q-value overestimation can be continually accumulated during learning, leading to sub-optimal policy updates and behaviors, causing instability and crippling the quality of learned policy \cite{lan_maxmin_2020}. In single-agent deep reinforcement learning (DRL), a lot of representative works have been proposed to address the overestimation problem, including Double DQN \cite{hasselt_deep_2016}, Averaged-DQN \cite{anschel_averaged-dqn_2017}, TD3 \cite{fujimoto_addressing_2018}, Minmax Q-learning \cite{lan_maxmin_2020}, and MeanQ \cite{liang_reducing_2022} etc. Although overestimation in single-agent DRL has been widely studied, overestimation in multiagent reinforcement learning (MARL) has received comparatively little attention although it increases with the number of agents and could cause severe learning instability to harm agent policy \cite{pan_regularized_2021}.

To mitigate overestimation in MARL, a few preceding works extend ensemble methods from single-agent DRL to reduce Q-value overestimation by discarding large target values in the ensemble. For example, Ackermann et al. \cite{ackermann_reducing_2019} introduce the TD3 technique to reduce the overestimation bias by using double centralized critics. Recently, Wu et al. \cite{wu_sub-avg_2022} use an ensemble of target multiagent Q-values to derive a lower update target by discarding the larger previously learned action values and averaging the retained ones. Besides the above ensemble-based methods, another category of works softens the maximum operator in the Bellman equation to avoid updating with large target values. For example, Gan et al. \cite{gan_stabilizing_2021} extend the soft Mellowmax operator into the field of MARL to alleviate the multiagent overestimation issue. At the same time, Pan et al. \cite{pan_regularized_2021} use the softmax Bellman operator on the joint action space to avoid large target Q-values and speed up the softmax computation by sampling actions around the maximal joint action. However, the above works focus on reducing the overestimation of MARL in the estimation process of the target Q-value. They ignore the follow-up optimization process of the online Q-network where overestimation can accumulate, thus making it hard to fully solve the complex multiagent overestimation problem.

To address this challenge, in our work, we first establish an iterative estimation-optimization framework to analyze the overestimation phenomenon in multiagent value-mixing Q-learning. Through the analysis, we found that multiagent overestimation not only comes from the overestimation of target Q-values \cite{gan_stabilizing_2021}, but also accumulates in the online Q-network after optimization with the overestimated update target. Inspired by this finding, we propose the \textbf{D}ual \textbf{E}nsembled \textbf{M}ulti\textbf{A}gent Q-learning with hypernet \textbf{R}egularizer (DEMAR) algorithm to address multiagent overestimation from two aspects. First, DEMAR extends a random ensemble technique \cite{chen_randomized_2021} into the estimation of target individual and global Q-values to derive a lower update target. Second, to prevent the accumulation of overestimation, we propose a novel hypernet regularizer to regularize the weights and biases from hypernetworks to constrain the optimization of online global Q-network. Extensive experiments in the classical multiagent particle environment and a noisy version of the StarCraft multiagent challenge environment show that DEMAR successfully stabilizes learning on various multiagent tasks suffering from overestimation while outperforming other baselines of the multiagent overestimation problem.

\section{Background}
\label{sec:background}
\subsection{Overestimation in Q-learning}
Q-learning \cite{watkins_q-learning_1992} learns the optimal value of each state-action via updating a tabular representation of the Q-value function by
\begin{equation}
    Q(s,a) \gets Q(s,a) + \alpha (r + \gamma \max_{a'} Q(s', a') - Q(s, a))
\end{equation}
to minimize the temporal difference error between estimates and bootstrapped targets with learning rate $\alpha$ and discount factor $\gamma$. DQN \cite{mnih_human-level_2015} learns a parametrized value function $Q^{\theta}$ by minimizing the squared error loss between $Q^{\theta}$ and the target value estimation as
\begin{equation}
    L(s,a,r,s';\theta)=(r + \gamma \max_{a'} Q^{\bar{\theta}}(s', a') - Q^{\theta}(s, a))^{2},
\end{equation}
where $Q^{\bar{\theta}}$ is a lagging version of the current state-action value function and is updated periodically from $Q^{\theta}$. When the target value estimation is with noise, the $\max$ operator has been shown to lead to an overestimation bias \cite{Thrun-1993-issues} due to Jensen's inequality as
\begin{equation}
    \mathbb{E}[\max_{a'} Q(s', a')] \geq \max_{a'}\mathbb{E}[Q(s', a')].
\end{equation}
Since the update is applied repeatedly through bootstrapping, it iteratively increases the bias of estimated Q-values and introduces instability into learning. 

\subsection{Randomized Ensembled Double Q-Learning}
\label{sec:redq}
Using an ensemble of Q-networks to reduce overestimation is widely studied in the field of single-agent DRL \cite{anschel_averaged-dqn_2017,chen_randomized_2021,liang_reducing_2022}. To reduce overestimation of the target Q-value, Double DQN \cite{hasselt_deep_2016} decomposes the max operation in the target into action selection and action evaluation, and TD3 \cite{fujimoto_addressing_2018} extends this idea into actor-critic methods. Meanwhile, Averaged-DQN \cite{anschel_averaged-dqn_2017} and MeanQ \cite{liang_reducing_2022} both employ an ensemble of multiple Q-networks to reduce overestimation with a lower estimation variance. 

Recently, Randomized Ensembled Double Q-Learning (REDQ) \cite{chen_randomized_2021} adopts an in-target minimization across a subset $\mathbb{H}$ of $N_{\mathbb{H}}$ Q-functions, which is randomly sampled from an ensemble of $H$ Q-functions, to derive a lower update target. Based on SAC \cite{haarnoja_soft_2018}, the Q target of REDQ is computed as
\begin{equation}
    y = r + \gamma (\min_{h \in \mathbb{H}} Q^{h}(s', \tilde{a}') - \alpha_{sac} \log \pi(s', \tilde{a}')), \tilde{a}' \sim \pi(\cdot | s'),
\end{equation}
where $h$ is the index of Q-function and $\alpha_{sac}$ is the coefficient of the entropy term in SAC. As indicated by Theorem 1 of REDQ \cite{chen_randomized_2021}, by adjusting $H$ and $N_{\mathbb{H}}$, the overestimation can be flexibly controlled. For example, by increasing the subset size $N_{\mathbb{H}}$ with some fixed ensemble size $H$, the estimation bias of the Q-value can be adjusted from above zero (overestimation) to under zero (underestimation).

\subsection{Multiagent Value-Mixing Q-learning}
We use Markov games, the multiagent extension of Markov Decision Processes \cite{littman_markov_1994}, as our setting. Markov games are described by a state transition function, $T: S \times A_{1} \times ... \times A_{N} \rightarrow P(S)$, which defines the probability distribution over all possible next states, $P(S)$, given the current global state $S$ and the action $A_{i}$ produced by the $i$-th agent. The reward is usually given based on the global state and actions of all agents $R_{i}:S \times A_{1} \times ... \times A_{N} \rightarrow \mathbb{R}$. If all agents receive the same rewards, i.e. $R_{1}=...=R_{N}$, Markov games are fully-cooperative \cite{matignon_independent_2012}: a best-interest action of one agent is also a best-interest action of other agents. Markov games can be partially observable and each agent $i$ receives a local observation $o_{i}: \mathcal{O}(S, i) \rightarrow O_{i}$. Thus, each agent learns a policy $\pi_{i}: O_{i} \rightarrow P(A_{i})$, which maps each agent's observation to a probability distribution over its action set, to maximize its expected discounted returns, $J_{i}(\pi_{i})=\mathbb{E}_{a_{1} \sim \pi_{1}, ..., a_{N} \sim \pi_{N},s \sim T}[\sum_{t=0}^{\infty} \gamma^{t} r_{i}(s_{t}, a_{1,t},..., a_{N,t})]$, where $\gamma$ is the discounted factor and is in the range of $[0,1)$.

To solve Markov games, multiagent value-mixing Q-learning algorithms \cite{rashid_qmix_2018,yang_qatten_2020} represent the global Q-value as an aggregation of individual Q-values with different forms. As the most representative method, QMIX \cite{rashid_qmix_2018} learns a monotonic mixing network to transform individual Q-values into the global Q-value as $Q_{tot}= f_{mix}(s, Q_{1}(s, a_{1}),..., Q_{N}(s, a_{N}))$. To ensure the monotonicity between global and individual Q-values that $\frac{\partial Q_{tot}}{\partial Q_{i}} \geq 0$, QMIX utilizes the hypernetworks \cite{ha_hypernetworks_2017} to produce non-negative weights and biases for mixing $Q_{i}$s. A hypernetwork $HPN$ is a network such as a multilayer perception. It takes the global state $s$ as the input to output non-negative weights and biases for the mixing network $f_{mix}$. For example, we can build $Q_{tot}$ with a linear form as $f_{mix}= \sum_{i=1}^{N} w_{i}Q_{i} + b$ where the weights $w_{i} \geq 0$ and bias $b$ are produced from the hypernetworks $\mathbf{w}=HPN_{1}(s)$ and $b=HPN_{2}(s)$ \cite{yang_qatten_2020}. The specific form of $f_{mix}$ in QMIX is described in Appendix B.
Recently, many works \cite{mahajan_maven_2019,mahajan_tesseract_2021,wang_qplex_2021} aim to address the monotonicity limitation of value-mixing MARL. The representation limitation from monotonicity causes inaccurate multiagent utilities (higher utilities on some joint-action pairs as overestimation and lower utilities on other pairs as underestimation) to lead to sub-optimal policies. Differently, the multiagent overestimation causes all joint-action multiagent Q-values larger than the ground truth to cripple policy quality. In this paper, we focus on the multiagent overestimation problem which receives much less attention nowadays.

\subsection{Overestimation in Multiagent Q-learning}
In the single-agent setting, if estimated action values contain independent noise uniformly distributed in $[-\epsilon, \epsilon]$ ($\epsilon > 0$), the expected overestimation of target values becomes
\begin{equation}
    \mathbb{E}[Z^{s}] = \mathbb{E}[r + \gamma \max_{a'}Q(s',a') - (r + \gamma \max_{a'} Q^{*}(s',a'))] \in [0, \gamma\epsilon\frac{m-1}{m+1}],
\end{equation}
where $Q^{*}$ is the target optimal Q-value and $m$ is the action space size \cite{Thrun-1993-issues}. Overestimation in MARL is more sophisticated as it increases with the number of agents \cite{pan_regularized_2021}. Furthermore, Gan et al. \cite{gan_stabilizing_2021} prove that the overestimation in multiagent value-mixing algorithms is raised from the maximization of individual Q-values with noise over its action space, which is shown below.
\begin{lemma}[Gan et al. \cite{gan_stabilizing_2021}]
\label{lemma:maoverestimation}
Let $Q_{tot}$ be a function of $s$ and $Q_{i}$ for $i=1,2,...,N$ where $Q_{i}$ is a function of $s$ and $a_{i}$ for $a_{i} \in A_{i}$. Assuming $l \leq \frac{\partial  Q_{tot}}{\partial Q_{i}} \leq L, i=1,2,...,N$ where $l \geq 0$, $L > 0$, and $Q_{i}(s,a_{i})$ is with an independent noise uniformly distributed in $[-\epsilon, \epsilon]$ on each $a_{i}$ given state $s$. Then
\begin{equation}
\label{eq:maoverestimation}
\begin{split}
    & lN\mathbb{E}[Z_{i}^{s}] \leq \mathbb{E}[r + \gamma \max_{\mathbf{a}'} Q_{tot}(s',\mathbf{Q}(s', \mathbf{a}'_{i})) - \\
    & (r + \gamma \max_{\mathbf{a}'} Q_{tot}(s',\mathbf{Q^{*}}(s', \mathbf{a}'_{i})))] \leq LN\mathbb{E}[Z_{i}^{s}],
\end{split}
\end{equation}
where $\mathbf{Q}(s', \mathbf{a}'_{i})= (Q_{1}(s',a'_{1}), ..., Q_{N}(s',a'_{N}))$ are individual Q-values, $\mathbf{Q^{*}}(s', \mathbf{a}'_{i})= (Q_{1}^{*}(s',a'_{1}), ..., Q_{N}^{*}(s',a'_{N}))$ are the target optimal individual Q-values, and $\mathbb{E}[Z_{i}^{s}] = \mathbb{E}[\max_{a'_{i}} Q_{i}(s',a'_{i}) - \max_{a'_{i}} Q_{i}^{*}(s',a'_{i})] \geq 0$. The proof is in Appendix A for completeness.
\end{lemma}
Eq.~(\ref{eq:maoverestimation}) shows that, when computing target global Q-values, overestimation is raised by maximizing $Q_{i}$ over its action space. However, previous works ignore the fact that the overestimation could also be accumulated when optimizing the Q-network, thus not fully tackling multiagent overestimation. Next, we show how multiagent overestimation is raised and accumulated during learning.

\section{DEMAR for Multiagent Overestimation}
In this section, first, we conduct an in-depth analysis of the overestimation in multiagent value-mixing Q-learning during the estimation-optimization iterations. Next, motivated by the findings of the analysis, we propose the Dual Ensembled Multiagent Q-Learning with Hypernet Regularizer algorithm as shown in Figure~\ref{fig:demar}.

\begin{figure}[htbp]
\centering
\includegraphics[width=1.0\linewidth]{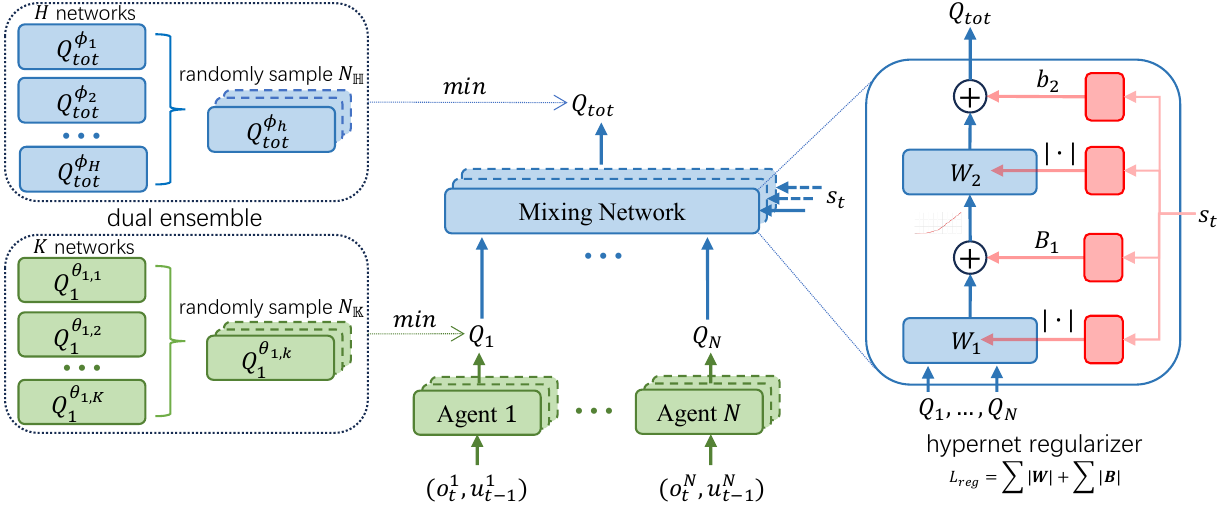}
\caption{The framework of DEMAR. The left part involves the dual ensembled multiagent Q-learning while the right part shows the hypernet regularizer on the global Q-network.}
\label{fig:demar}
\end{figure}

\subsection{Analysis of Multiagent Overestimation}
\label{subsection:analysis}

Multiagent value-mixing Q-learning learns in estimation-optimization iterations. At each iteration, the target global Q-value is estimated first. Then the online global Q-network is optimized based on the estimated target global Q-value. Following multiagent value-mixing Q-learning \cite{sunehag_value-decomposition_2018,rashid_qmix_2018,yang_qatten_2020} which approximates the global Q-value $Q_{tot}$ by mixing individual Q-values $Q_{i}$, we expand $Q_{tot}$ as
\begin{equation}
\label{eq_qmix}
    Q_{tot}=Q_{tot}(s,Q_{1},...,Q_{N}) = f_{mix}(s,Q_{1}(s,a_{1}),...,Q_{N}(s,a_{N})),
\end{equation}
where $f_{mix}$ is monotonic mixing network by enforcing $\frac{\partial Q_{tot}}{\partial Q_{i}} \geq 0$.

First, we analyze the estimation process of the target global Q-value. Lemma~\ref{lemma:maoverestimation} indicates that the overestimation of the target $Q_{tot}$ results from overestimated target $Q_{i}$s by monotonicity in multiagent value-mixing Q-learning. And the overestimation of the target $Q_{i}$ results from independent noise on each action-value of $Q_{i}$ by maximizing $Q_{i}$ over its action space. After an overestimated target $Q_{tot}$ is obtained, the online global Q-network is trained to fit this overestimated update target. Therefore, the overestimation of the target $Q_{tot}$ and $Q_i$ contribute to the multiagent overestimation.

Second, we analyze the optimization process of the online global Q-network. The update target of online global Q-network is computed as $y_{tot} = r + \gamma \max Q_{tot}(s',\mathbf{a}')$ with the overestimated target $Q_{tot}(s',\mathbf{a}')$. We found that, after the online global Q-network's optimization with $y_{tot}$, $\frac{\partial Q_{tot}}{\partial Q_i}$ impacts the overestimation of the global Q-network with a quadratic term as below.

\begin{theorem}
\label{severetheorem}
Assume that the current $Q_{tot}$ network approximates the optimal $Q_{tot}^*(\cdot)$ as $Q_{tot}(\cdot)=Q_{tot}^*(\cdot)$ and the current $Q_{i}$ network approximates the optimal $Q_{i}^*(\cdot)$ as $Q_{i}(\cdot)=Q_{i}^*(\cdot)$. By Lemma~\ref{lemma:maoverestimation}, if $l \leq \frac{\partial  Q_{tot}}{\partial Q_{i}} \leq L$ for $i=1, 2,..., N$ where $l \geq 0$, $L > 0$ and the estimated individual Q-value is with an independent noise uniformly distributed in $[-\epsilon, \epsilon]$ on each action, then the estimated update target $y_{tot}$ will be biased from the optimal target value $y^*_{tot}$ as $y_{tot} = y^*_{tot} + \Delta y$ where the estimation bias $\Delta y > 0$. If we continue to train the current $Q_{i}$ network by the loss $L_{mix}=\|y_{tot} - Q_{tot}\|^2$, the updated network $\hat{Q}_{i}$ will be biased from optimal $Q^{*}_i$ as
\begin{equation}
    \hat{Q}_{i} = Q^{*}_i + \Delta Q_i,
\end{equation}
where $\Delta Q_i \geq 2 \alpha l\Delta y $ for some learning rate $\alpha>0$. After the feedforward even without considering the updated $Q_{tot}$ network, the bias of the new $Q_{tot}$ value $\hat{Q}_{tot}$ from the optimal value $Q_{tot}^*$ becomes
\begin{equation}
\label{eq:severeoverestimation}
    \hat{Q}_{tot} = Q_{tot}^* + 2\alpha\Delta y\sum_{i=1}^N (\frac{\partial Q_{tot}}{\partial Q_i})^2.
\end{equation}
\end{theorem}

\begin{proof}
From Lemma~\ref{lemma:maoverestimation}, we have an overestimated target value $y_{tot}=y_{tot}^{*} + \Delta y$ where $\Delta y > 0$. As we assume that the learned $Q_{tot}$ network approximates the optimal function, we have $Q_{tot}(s,Q_1,\ldots,Q_N) = y_{tot}^{*}$.
We apply the gradient method to minimize $L_{mix}$, the independent $Q_i$ is updated as follows,
\begin{equation}
\begin{split}
    \hat{Q}_i & = Q^*_i -\alpha \frac{\partial L_{mix}}{\partial Q_i}\\
    & = Q^*_i -\alpha \frac{\partial(y_{tot} - Q_{tot}(s,Q_1,\ldots,Q_N))^2}{\partial Q_i}\\
    & = Q^*_i - \alpha\frac{\partial(y^*_{tot}+\Delta y - Q_{tot}(s,Q_1,\ldots,Q_N))^2}{\partial Q_i}\\
    & = Q^*_i + 2\alpha(y^*_{tot}+\Delta y - Q_{tot}(s,Q_1,\ldots,Q_N))\frac{\partial Q_{tot}(s,Q_1,\ldots,Q_N)}{\partial Q_i} \\
    & = Q^*_i + 2\alpha\Delta y\frac{\partial Q_{tot}(s,Q_1,\ldots,Q_N)}{\partial Q_i}.
\end{split}
\end{equation}
Compared with updating with the ground-truth target global Q-value $y_{tot}^*$, $Q_i$ is biased with $\Delta Q_i = 2\alpha\Delta y\frac{\partial Q_{tot}}{\partial Q_i} \geq 2 \alpha l\Delta y$. Such a bias $\Delta Q_{i}$ will be propagated through the global Q-network's feedforward process to increase the overestimation of the new global Q-value $\hat{Q}_{tot}$ as
\begin{equation}
\label{eq:taylorexpansion}
    \begin{split}
        \hat{Q}_{tot}&=Q_{tot}(s,\hat{Q}_1,\ldots,\hat{Q}_N)
        =Q_{tot}(s,Q_1+\Delta Q_1,\ldots,Q_N+\Delta Q_N)\\
        &\approx Q_{tot}(s,Q_1,\ldots,Q_N) + \sum_{i=1}^N \frac{\partial Q_{tot}(s,Q_1,\ldots,Q_N)}{\partial Q_i}\Delta Q_i \\
        &=Q_{tot}(s,Q_1,\ldots,Q_N) + 2\alpha\Delta y\sum_{i=1}^N (\frac{\partial Q_{tot}(s,Q_1,\ldots,Q_N)}{\partial Q_i})^2 \\
        &= Q_{tot}^* + 2\alpha\Delta y\sum_{i=1}^N (\frac{\partial Q_{tot}(s,Q_1,\ldots,Q_N)}{\partial Q_i})^2 \\
        &\geq Q^*_{tot}+2\alpha N l^2\Delta y,
    \end{split}
\end{equation}
where the second approximation comes from the first-order Taylor expansion of the multivariate function. Thus, $Q_{i}$'s overestimation causes more severe overestimation for $Q_{tot}$ even if we only update the online $Q_{i}$ network in one optimization step. Furthermore, in repeated estimation-optimization iterations, such a bias accumulates to increase the multiagent Q-value overestimation.
\end{proof}

As Eq.~(\ref{eq:severeoverestimation}) shows, when forwarding the updated biased $\hat{Q}_{i}$ to compute a new global Q-value $\hat{Q}_{tot}$, $\frac{\partial Q_{tot}}{\partial Q_i}$ impacts the global Q-value's overestimation with a quadratic term. Besides, Pan et al. \cite{pan_regularized_2021} empirically show that $\frac{\partial Q_{tot}}{\partial Q_i}$ in QMIX continually increases with overestimation. Therefore, we also need to regularize $\frac{\partial Q_{tot}}{\partial Q_i}$ in $Q_{tot}$'s optimization. In conclusion, the overall overestimation analysis motivates us to control the overestimated $Q_{tot}$ and $Q_{i}$ in the target Q-value estimation, as well as to regularize $\frac{\partial Q_{tot}}{\partial Q_{i}}$ in the online $Q_{tot}$ network's optimization.

\begin{algorithm*}[htbp]
\small
\caption{Dual Ensembled Multiagent Q-Learning with Hypernet Regularizer (DEMAR)}
\label{alg:demar}
\begin{algorithmic}[1]
    \STATE Initialize individual Q-value network parameters $\theta_{1,1},..., \theta_{1,K}, ..., \theta_{N,1}, ..., \theta_{N,K}$, global Q-value network parameters $\phi_{1}, \phi_{2}, ..., \phi_{H}$ and empty replay buffer $D$. Set target network parameters $\bar{\theta}_{i,k} \gets \theta_{i,k}$ for $i=1,2,...,N$, $k=1,2,...,K$ and $\bar{\phi}_{h} \gets \phi_{h}$ for $h=1,2,...,H$.
    \FOR{Episode $1,2,3...$}
    \STATE Each agent $i$ takes action $a_{i,t} \sim \pi_{\theta_{i}}(\cdot|o_{i,t})$. Step into state $s_{t+1}$. Receive reward $r_{t}$ and observe $o_{i,t+1}$.
    \STATE Add data to buffer: $D \gets D \cup \{(s_{t}, \mathbf{o}_{t}, \mathbf{a}_{t}, r_{t}, s_{t+1}, \mathbf{o}_{t+1})\}$.
    \STATE Sample a mini-batch $B=\{(s,\mathbf{o},\mathbf{a}_{t},r,s',\mathbf{o}')\}$ from $D$.
    \STATE Sample a set $\mathbb{K}$ of $N_{\mathbb{K}}$ distinct indices from $\{1, 2, ..., K\}$.
    \STATE Sample a set $\mathbb{H}$ of $N_{\mathbb{H}}$ distinct indices from $\{1, 2, ..., H\}$.
    \STATE Compute the Q target $y_{tot}$ which is the same for all $H$ critics (denote $\bar{\theta}_{i,1}, ..., \bar{\theta}_{i,K}$ as $\bar{\boldsymbol{\theta}}_{i}$):
    \begin{equation}
    \begin{split}
    \nonumber
        y_{tot} = r + \gamma & ({\color{blue}\min_{h \in \mathbb{H}}Q_{tot}^{\bar{\phi}_{h}}}(s',Q_{1}^{\bar{\boldsymbol{\theta}}_{1}}(o_{1}',a_{1}'),...,Q_{N}^{\bar{\boldsymbol{\theta}}_{N}}(o_{N}',a_{N}')), \text{where } Q_{i}^{\bar{\boldsymbol{\theta}}_{i}}(o_{i}',a_{i}') = \max_{a_{i}'}{\color{blue}\min_{k \in \mathbb{K}} Q_{i}^{\bar{\theta}_{i,k}}}(o_{i}',a_{i}').
    \end{split}
    \end{equation}
    \FOR{$h=1,...,H$}
    \STATE Update $\phi_{h}, \theta_{1,1},...,\theta_{N,K}$ (denote $\theta_{1,1},...,\theta_{N,K}$ as $\boldsymbol{\theta}$) with gradient descent:
    \begin{equation}
    \begin{split}
    \nonumber
        & \quad \nabla_{\phi_{h}, \boldsymbol{\theta}} L(\phi_{h}, \boldsymbol{\theta}), \text{where}\\ 
        L(\phi_{h},\boldsymbol{\theta}) = \frac{1}{|B|} \sum_{(s,\mathbf{o},\mathbf{a},r,s',\mathbf{o}') \in B} ([Q^{\phi_{h}}_{tot}(s,Q_{1}^{\boldsymbol{\theta}_{1}}&(o_{1},a_{1}),...,Q_{N}^{\boldsymbol{\theta}_{N}}(o_{N},a_{N})) - y_{tot}]^{2} + {\color{blue}\alpha_{reg} L_{reg}^{h}(s)}), \text{and } Q_{i}^{\boldsymbol{\theta}_{i}}(o_{i},a_{i})={\color{blue}mean_{k=1}^{K} Q_{i}^{\theta_{i,k}}}(o_{i},a_{i}).
    \end{split}
    \end{equation}
    \ENDFOR
    \STATE Update target networks $\bar{\phi}_{1} \gets \phi_{1}$, ...,  $\bar{\phi}_{H} \gets \phi_{H}$ and $\bar{\theta}_{1,1} \gets \theta_{1,1}$, ..., $\bar{\theta}_{N,K} \gets \theta_{N,K}$ every $C$ times.
    \ENDFOR
\end{algorithmic}
\end{algorithm*}

\begin{figure*}[htbp]
\centering
\subfigure[simple\_tag]{
\label{tag}
\includegraphics[width=0.321\textwidth]{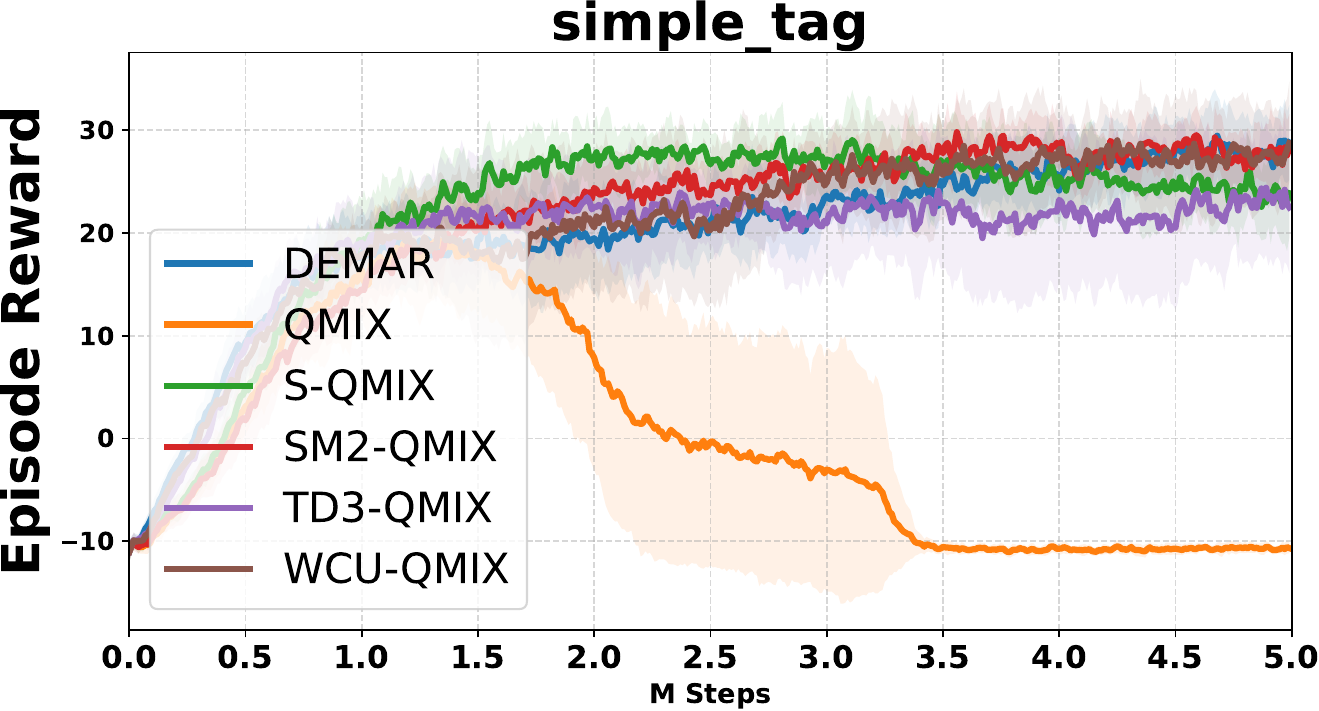}}
\subfigure[simple\_world]{
\label{world}
\includegraphics[width=0.321\textwidth]{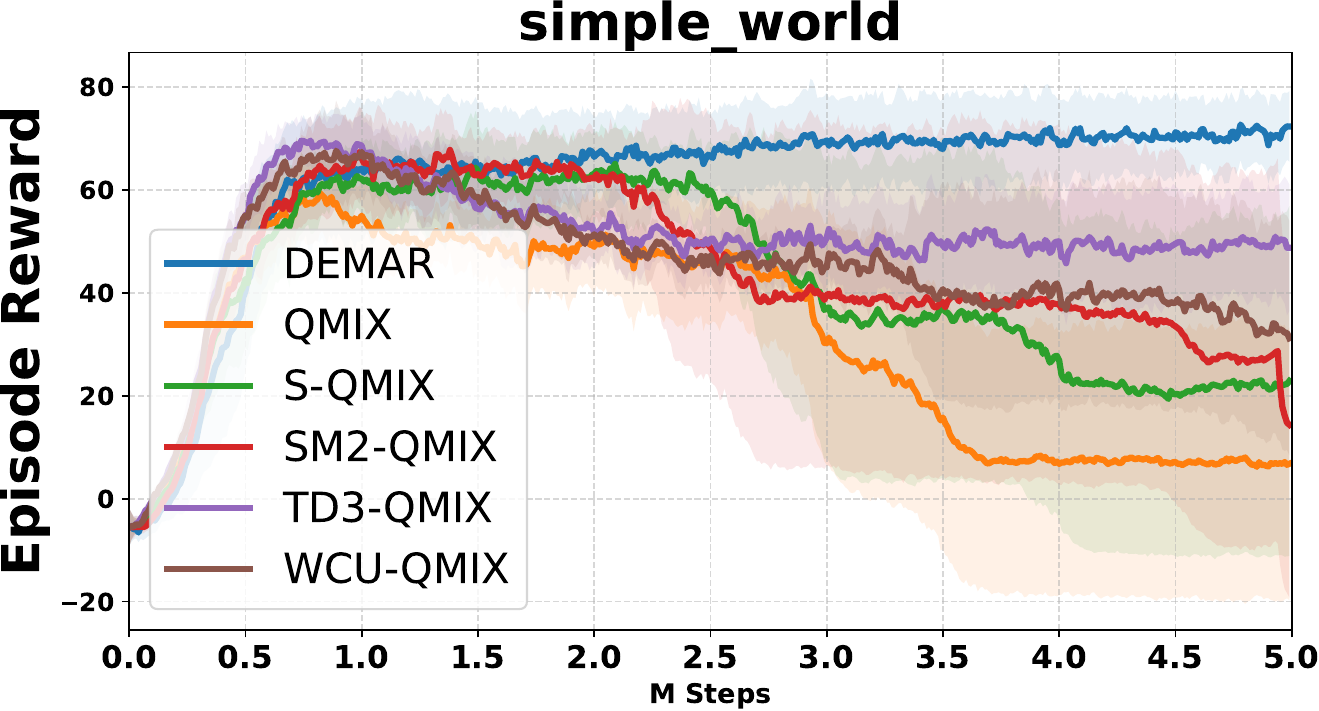}}
\subfigure[simple\_adversary]{
\label{adversary}
\includegraphics[width=0.321\textwidth]{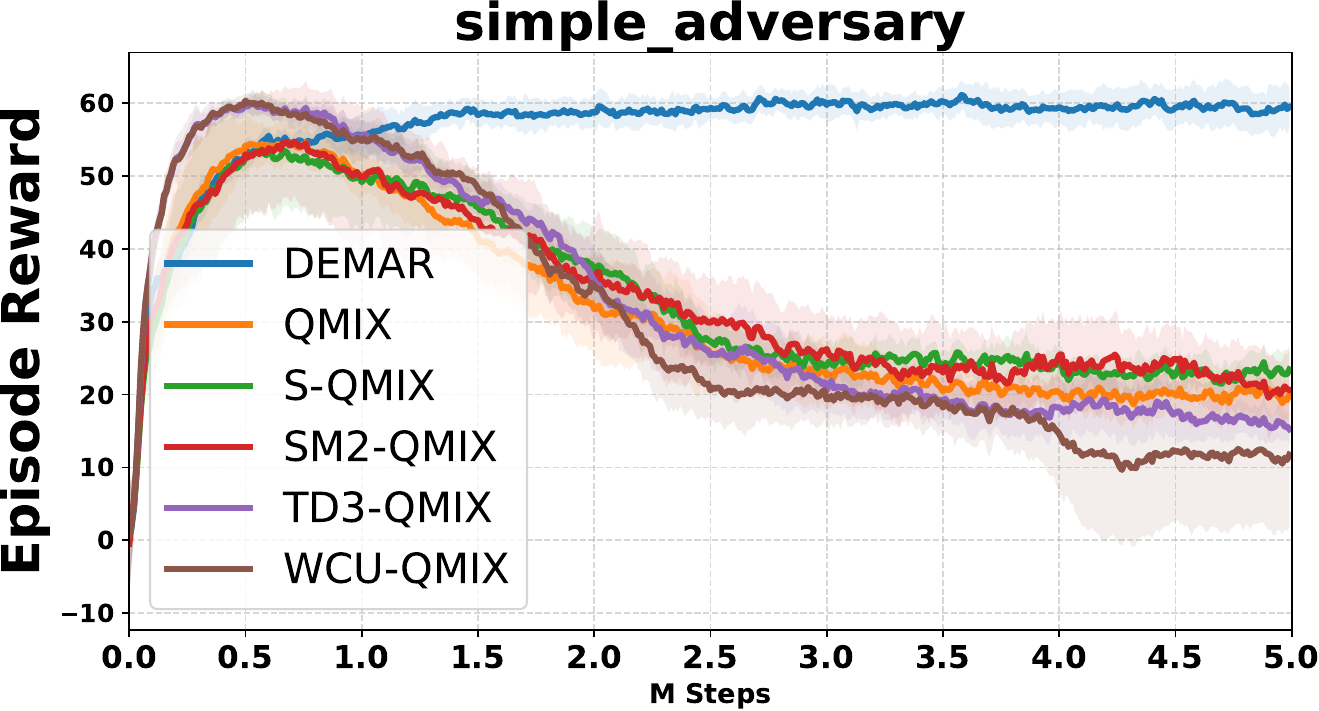}}
\subfigure[simple\_tag overestimation]{
\label{tagoverestimation}
\includegraphics[width=0.321\textwidth]{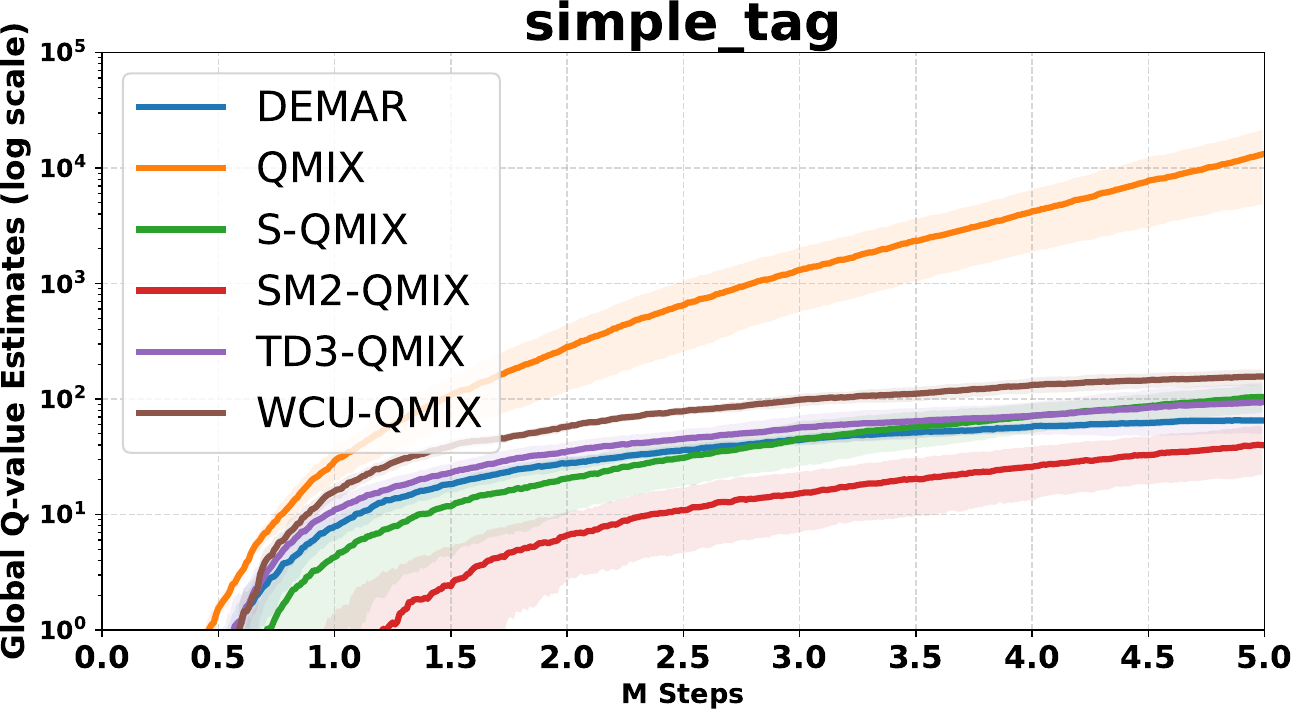}}
\subfigure[simple\_world overestimation]{
\label{worldoverestimation}
\includegraphics[width=0.321\textwidth]{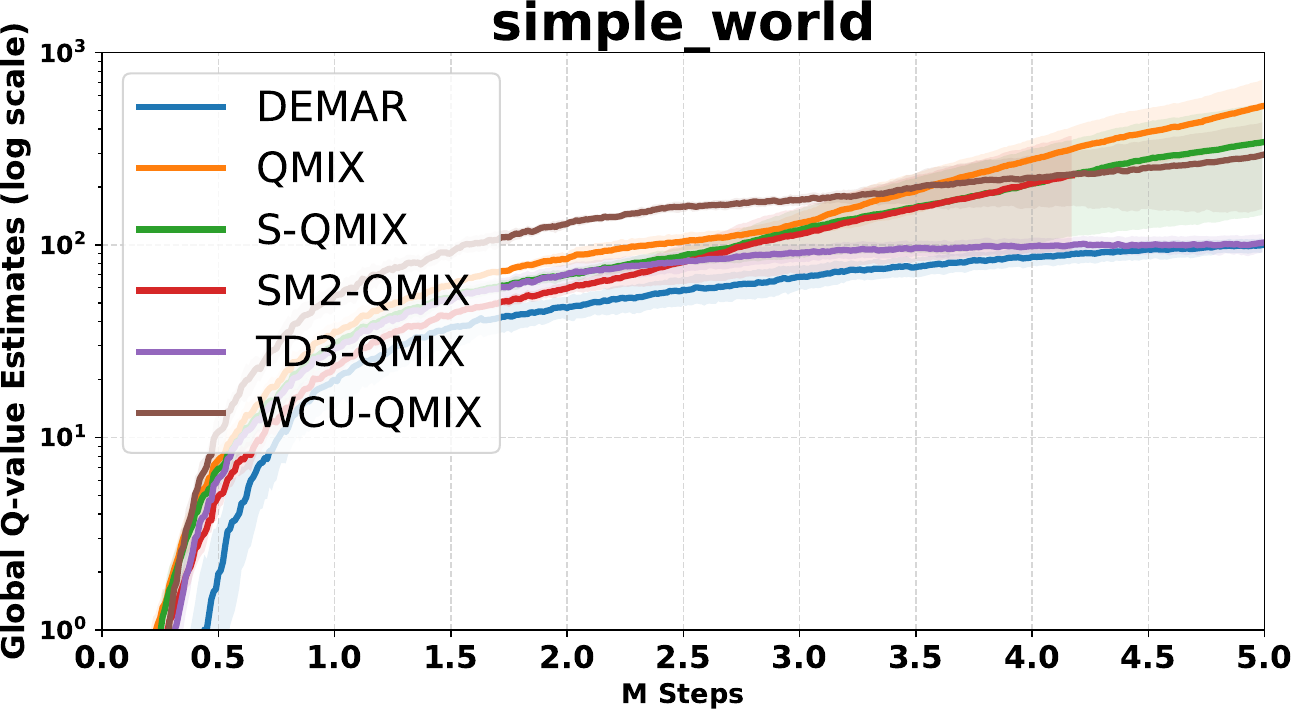}}
\subfigure[simple\_adversary overestimation]{
\label{adversaryoverestimation}
\includegraphics[width=0.321\textwidth]{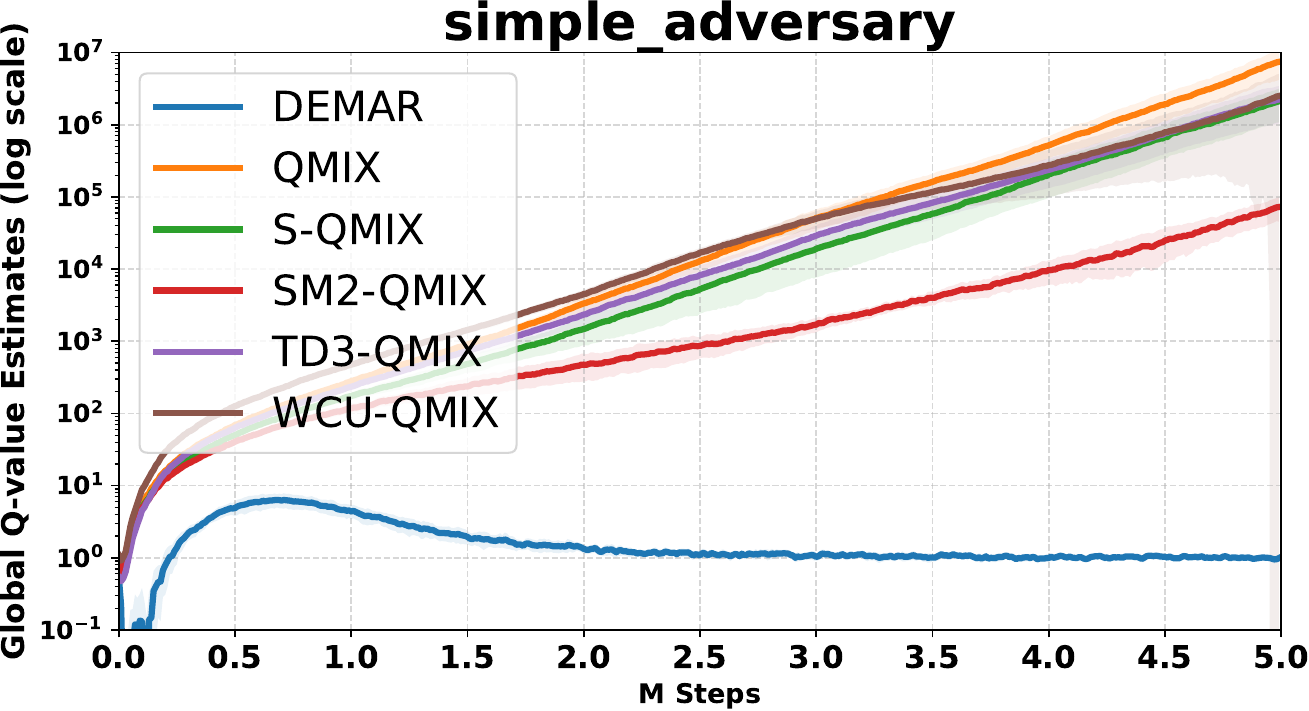}}
\caption{Results on different MPE scenarios. Figure~\ref{tag}-\ref{adversary} show the learning performance of each method on MPE tasks. Figure~\ref{tagoverestimation}-\ref{adversaryoverestimation} show the estimated global Q-value of each method in the log scale on MPE tasks.}
\label{figure:mpe}
\end{figure*}

\begin{figure*}[htbp]
\centering
\subfigure[5m\_vs\_6m]{
\label{5m6m}
\includegraphics[width=0.242\textwidth]{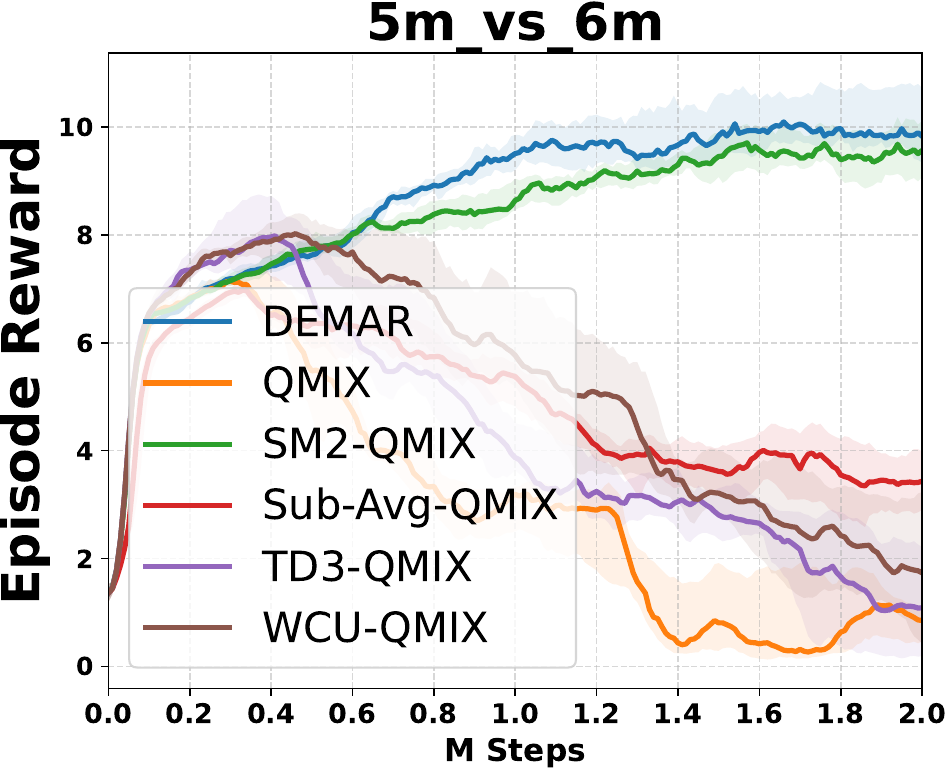}}
\subfigure[2s3z]{
\label{2s3z}
\includegraphics[width=0.242\textwidth]{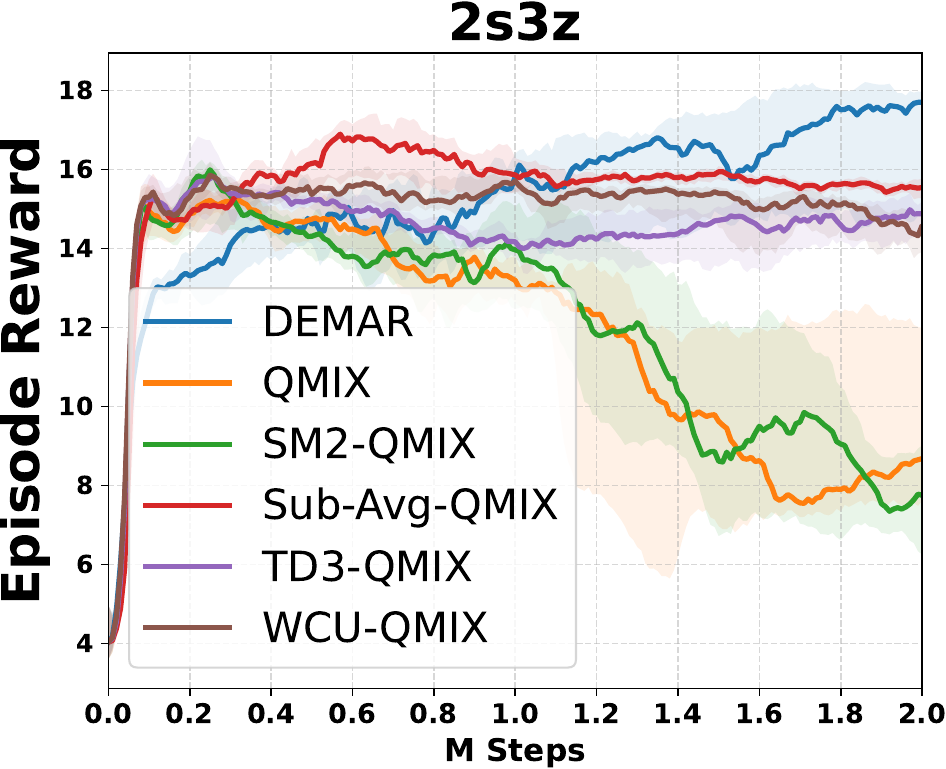}}
\subfigure[3s5z]{
\label{3s5z}
\includegraphics[width=0.242\textwidth]{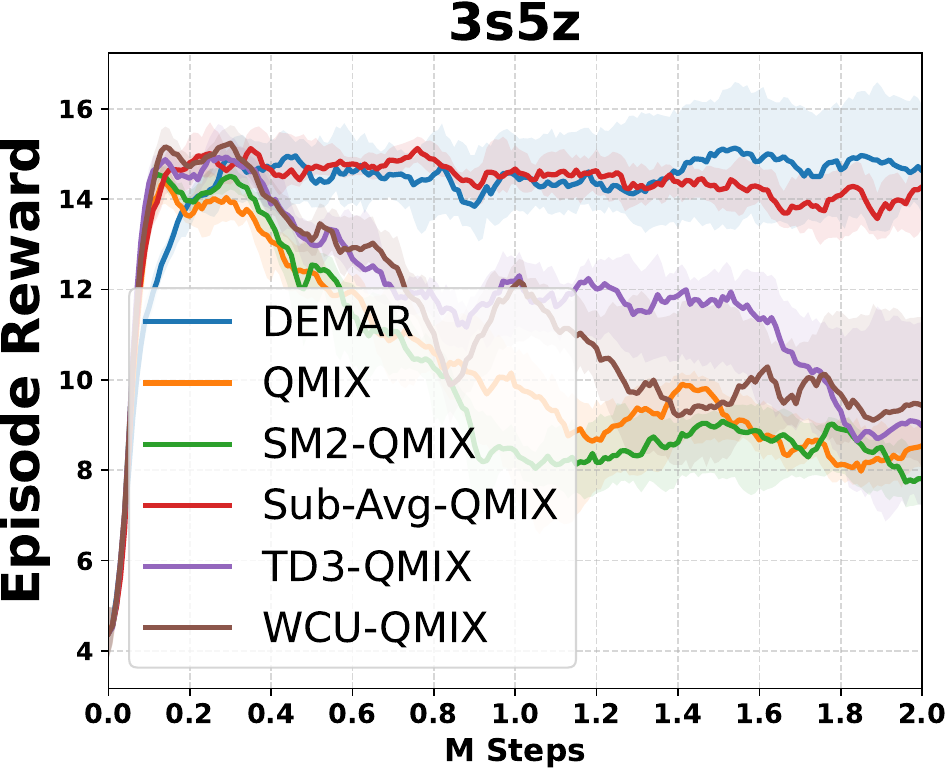}}
\subfigure[10m\_vs\_11m]{
\label{10m11m}
\includegraphics[width=0.242\textwidth]{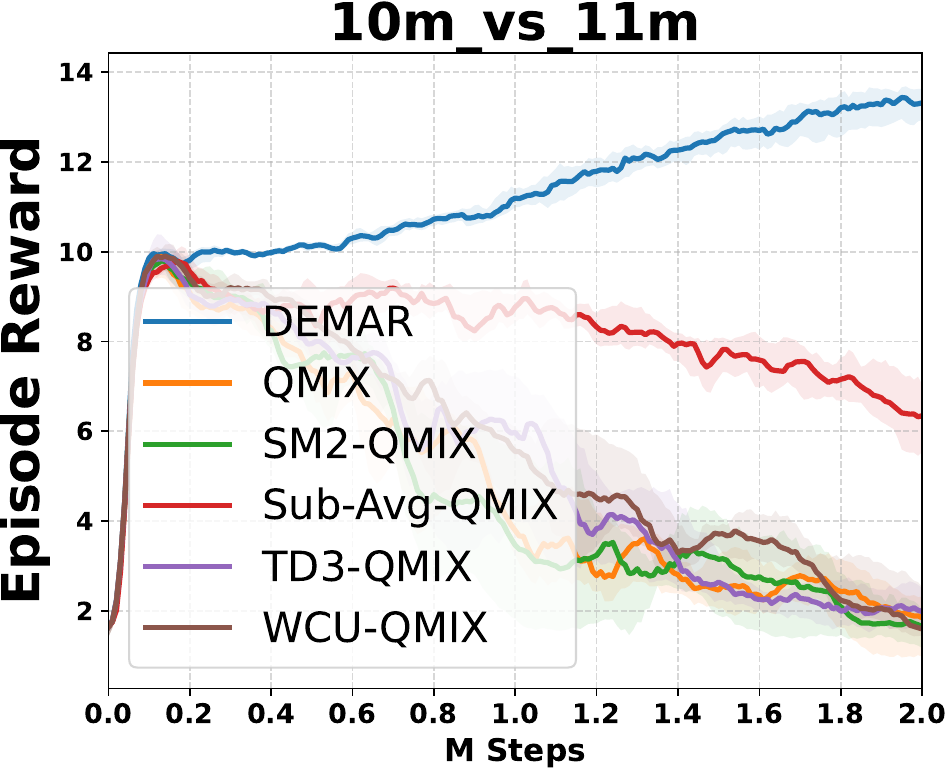}}
\subfigure[$Q_{tot}$ on 5m\_vs\_6m]{
\label{5m6moverestimation}
\includegraphics[width=0.242\textwidth]{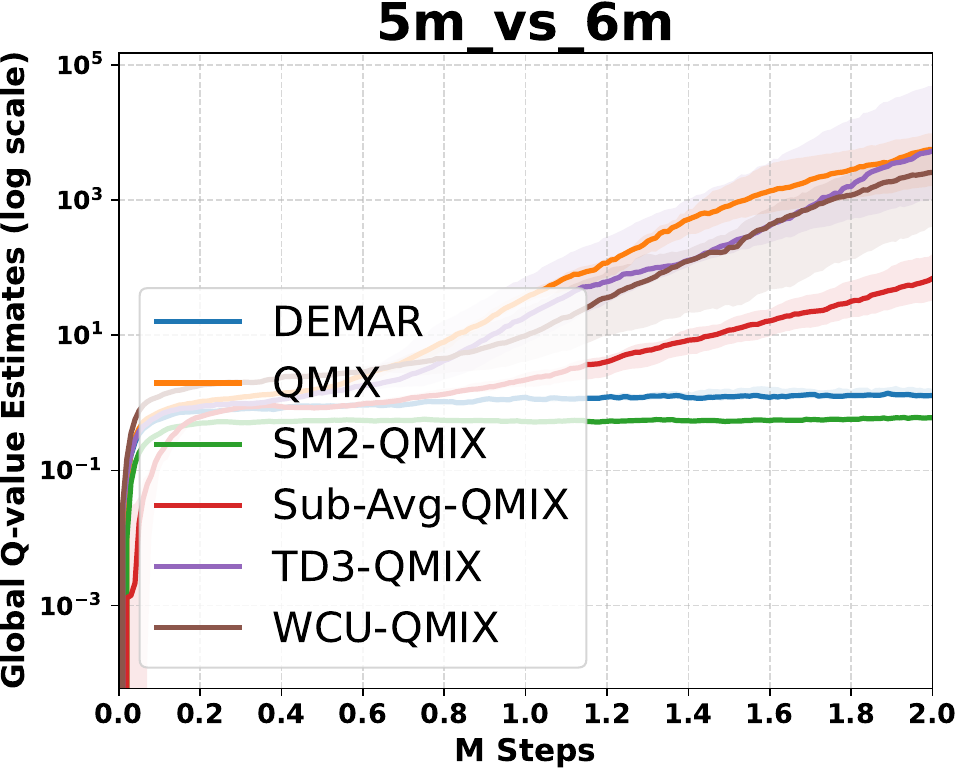}}
\subfigure[$Q_{tot}$ on 2s3z]{
\label{2s3zoverestimation}
\includegraphics[width=0.242\textwidth]{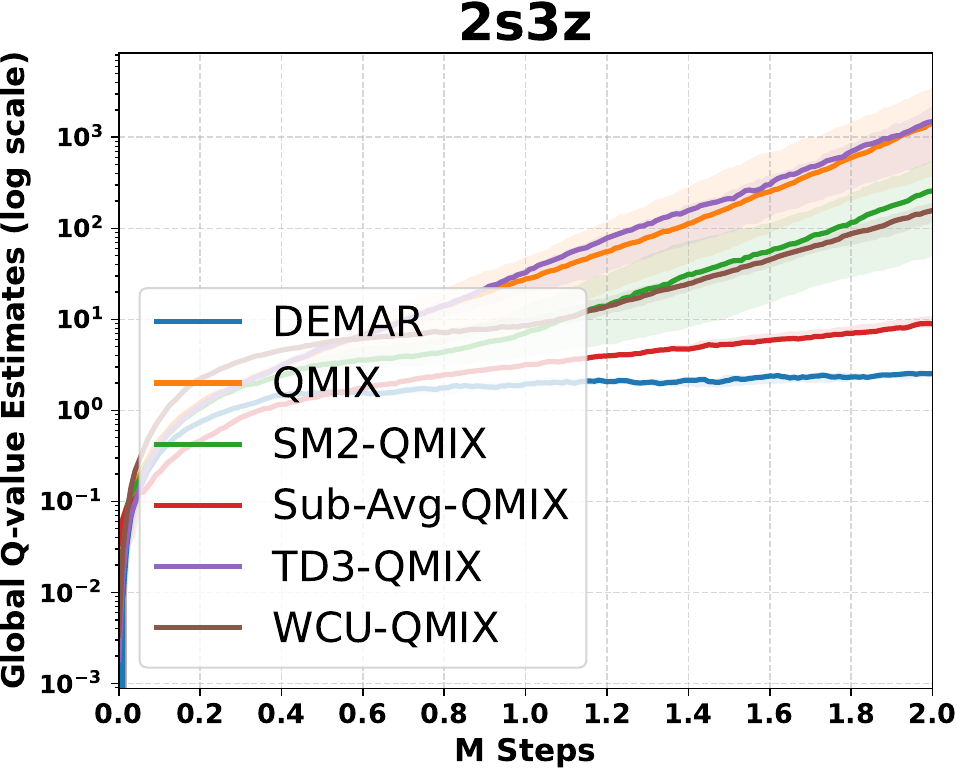}}
\subfigure[$Q_{tot}$ on 3s5z]{
\label{3s5zoverestimation}
\includegraphics[width=0.242\textwidth]{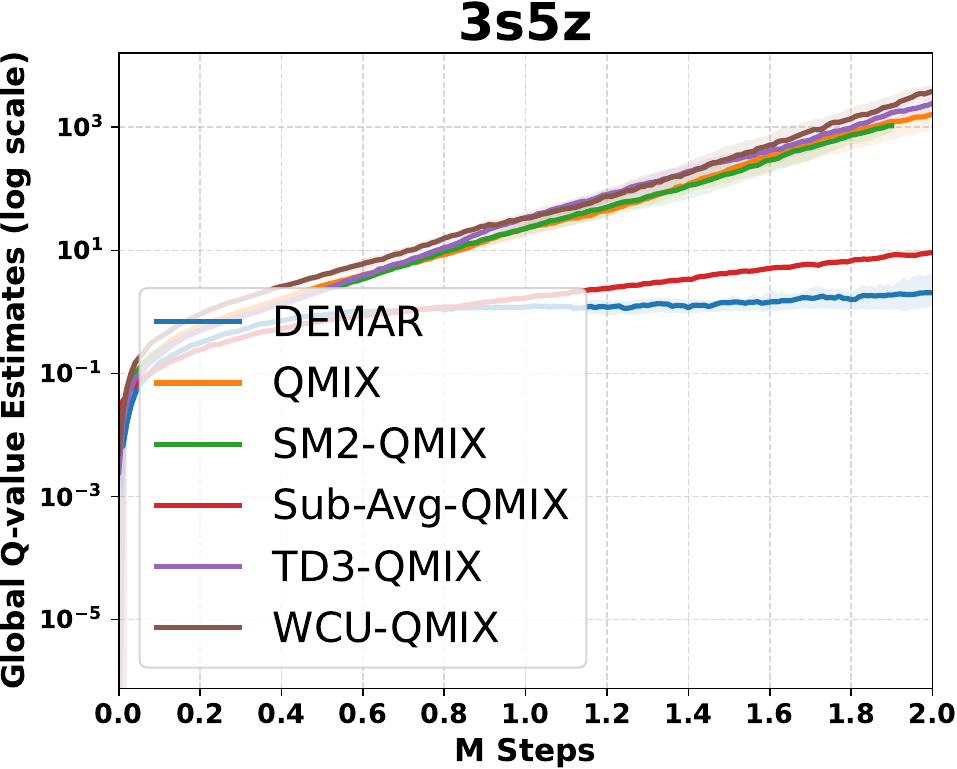}}
\subfigure[$Q_{tot}$ on 10m\_vs\_11m]{
\label{10m11moverestimation}
\includegraphics[width=0.242\textwidth]{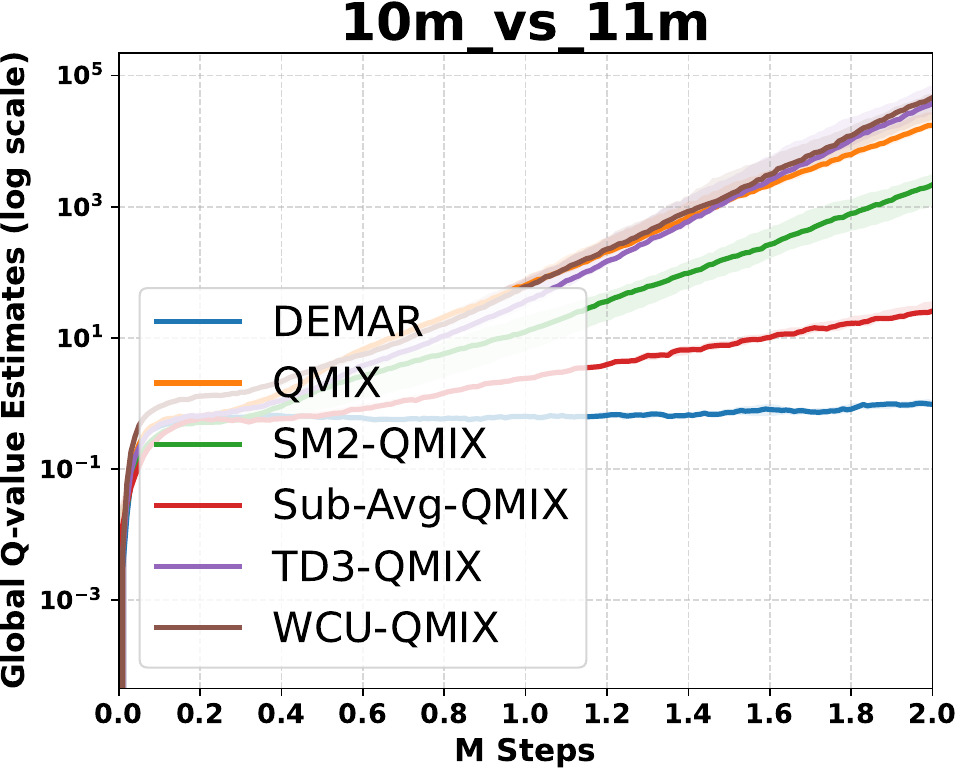}}
\caption{Results on different noisy SMAC scenarios. Figure~\ref{5m6m}-\ref{10m11m} show the learning performance of each method on SMAC tasks. Figure~\ref{5m6moverestimation}-\ref{10m11moverestimation} show the estimated global Q-value of each method in the log scale on SMAC tasks.}
\label{figure:smac}
\end{figure*}

\subsection{Dual Ensembled Multiagent Q-Learning}
\label{subsection:dema}
Motivated by the above analysis, first, we are going to control the overestimated $Q_{tot}$ and $Q_{i}$ in the estimation process of target $Q_{tot}$. Similar to works in single-agent DRL, here we introduce the idea of the ensemble into the estimation process of target Q-values to derive a lower update target. We extend REDQ \cite{chen_randomized_2021}, a state-of-the-art ensemble method with theoretical guarantee and impressive performance from single-agent DRL, into the multiagent value-mixing Q-learning algorithms. As shown in Section~\ref{sec:redq}, REDQ uses an in-target minimization across a random subset of Q-functions from the ensemble of Q-networks to reduce overestimation. However, REDQ is specially built on the single-agent SAC \cite{haarnoja_soft_2018} algorithm. It cannot be directly applied to the multiagent value-mixing Q-learning which involves the process of mixing individual $Q_{i}$s into $Q_{tot}$. In this study, we carefully design a dual ensembled algorithm based on the random ensemble technique from REDQ to control both the overestimated target $Q_{i}$ and $Q_{tot}$ during the mixing process. Next, we explain the details.

The global target value of multiagent value-mixing Q-learning \cite{rashid_qmix_2018} is computed as 
\begin{equation}
\label{eq:targetqtot}
\begin{aligned}
    y_{tot} & = r + \gamma \max_{\mathbf{a}'} Q_{tot}^{\bar{\phi}}(s',\mathbf{Q}(s',\mathbf{a}'_{i})) \\
    & = r + \gamma Q_{tot}^{\bar{\phi}}(s', \max_{\mathbf{a}'} \mathbf{Q}(s',\mathbf{a}'_{i})) \\
    & = r + \gamma Q_{tot}^{\bar{\phi}}(s',\max_{a_{1}'}Q^{\bar{\theta}_{1}}_{1}(o_{1}',a_{1}'),...,\max_{a_{N}'}Q^{\bar{\theta}_{N}}_{N}(o_{N}',a_{N}')),
\end{aligned}
\end{equation}
where $\bar{\phi}$ and $\bar{\theta}$ are the parameters of the target global Q-network and target individual Q-networks respectively. As analyzed in Section~\ref{subsection:analysis}, the overestimation in target $Q_{i}$ results in the overestimation of target $Q_{tot}$. Therefore, we first apply the minimization operation in the random ensemble of target $Q_{i}$ networks to reduce target $Q_{i}$'s overestimation. Thus, the target $Q_{i}$ of agent $i$ is computed as
\begin{equation}
    Q_{i}^{\bar{\boldsymbol{\theta}}_{i}} = \min_{k \in \mathbb{K}} Q_{i}^{\bar{\theta}_{i,k}},
\end{equation}
where $\mathbb{K}$ is a subset with size $N_{\mathbb{K}}$ randomly sampled from $\{1,2,...,K\}$ and $\bar{\theta}_{i,k}$ are the parameters of the agent $i$'s $k$th target individual Q-network. After getting the target $Q_{i}$, we are able to compute the target $Q_{tot}$. Similarly, we use the minimization operation again in the random ensemble of target $Q_{tot}$ networks to reduce the overestimation of target global Q-value as
\begin{equation}
    Q_{tot}^{\bar{\phi}} = \min_{h \in \mathbb{H}} Q_{tot}^{\bar{\phi}_{h}},
\end{equation}
where $\mathbb{H}$ is a subset with size $N_{\mathbb{H}}$ randomly sampled from $\{1,2,...,H\}$ and $\bar{\phi}_{h}$ are the parameters of the $h$th target global Q-network. With the proposed dual ensembled Q-learning technique, we derive a lower update target for the online global Q-network. In addition, as indicated by Theorem 1 of REDQ \cite{chen_randomized_2021}, we are able to flexibly control the overestimation of the target $Q_{tot}$ and $Q_{i}$ by changing their Q-network ensemble sizes and the random subset sizes respectively. Next, we are going to prevent the overestimation accumulation in the optimization process of the global Q-network.

\subsection{Hypernet Regularizer}
\label{subsection:hr}
The analysis in Section~\ref{subsection:analysis} indicates that, in the optimization step of the online $Q_{tot}$ network, $\frac{\partial Q_{tot}}{\partial Q_{i}}$ impacts the multiagent overestimation with a quadratic term. To tackle this issue, we propose a novel hypernet regularizer to constrain this term in the online global Q-network's optimization. Specifically, we use the L1 sparse regularization on hypernetwork weights and biases to constrain $\frac{\partial Q_{tot}}{\partial Q_{i}}$ as
\begin{equation}
    L_{reg}=\sum{|\mathbf{W}_f|} + \sum{|\mathbf{B}_f|},
\end{equation}
where $\mathbf{W}_f$ and $\mathbf{B}_f$ are hypernetwork weights and biases in $f_{mix}$ produced from separate hypernetworks. The proof of using the hypernet regularizer to constrain $\frac{\partial Q_{tot}}{\partial Q_{i}}$ is given in Appendix B.

The final loss function for DEMAR becomes
\begin{equation}
    L(\phi_{h},\boldsymbol{\theta})=L_{mix}^{h} + \alpha_{reg} L_{reg}^{h},
\end{equation}
where $\alpha_{reg}$ is the coefficient of hypernet regularization. $h$ is the index of $Q_{tot}$ network, $\phi_{h}$ are the network parameters of $h$th $Q_{tot}$ network, and $\boldsymbol{\theta}$ are the network parameters of all $Q_{i}$ networks.

DEMAR is completely described in Algorithm~\ref{alg:demar}. \textbf{Line 1} initializes the empty replay buffer, the online and target individual Q-value networks for each agent, as well as the online and target global Q-value networks. \textbf{Line 3-4} interact with the environment based on agent policies and store the transition into the replay buffer. \textbf{Line 5} samples a mini-batch of transitions from the replay buffer. \textbf{Line 6-7} sample the indices for selecting instances from the $Q_{i}$ network ensemble and $Q_{tot}$ network ensemble respectively. \textbf{Line 8} computes the target $Q_{tot}$ value with in-target minimization across the selected network subsets. Specifically, for the computation of $\max_{a_{i}'}\min_{k \in \mathbb{K}} Q_{i}^{\bar{\theta}_{i,k}}(o_{i}', a_{i}')$, the $\min_{k \in \mathbb{K}} Q_{i}^{\bar{\theta}_{i,k}}$ is first computed for each action in agent $i$' action space by finding the minimal value in the ensemble subset, then the $\max_{a_{i}'}$ operation is executed. \textbf{Line 9-11} update all online $Q_{i}$ and $Q_{tot}$ networks by the loss $L$. Note that the online $Q_{i}$ is calculated by averaging the ensemble of $Q_{i}$ values with given actions instead of randomly picking up one value from the $Q_{i}$ ensemble, which enjoys the benefit of reducing value variance. \textbf{Line 12} updates all target $Q_{i}$ and $Q_{tot}$ networks by copying network parameters from their online versions periodically. With the dual ensembled Q-learning and hypernet regularizer, DEMAR is able to control the multiagent overestimation. Additionally, by setting $H=N_{\mathbb{H}}=K=N_{\mathbb{K}}=1$ and $\alpha_{reg}=0$, DEMAR degenerates to vanilla QMIX.

\section{Experiments}
\label{sec:experiments}

In this section, we first conduct experiments in the multiagent particle environment (MPE) \cite{lowe_multi-agent_2017} and a noisy version of the StarCraft multiagent challenge (SMAC) \cite{samvelyan_starcraft_2019} that both suffer from overestimation. Then we perform ablation studies to validate the proposed dual ensembled Q-learning and hypernet regularizer of DEMAR separately. Next, we experimentally examine the overestimation terms as analyzed in Section~\ref{subsection:analysis} and compare the estimated Q-value of DEMAE with its true Q-value. Finally, we extend DEMAR to other advanced MARL approaches to test the generality of DEMAR.

\subsection{Experiments in MPE}

We first evaluate DEMAR in MPE including the \textit{simple\_tag}, \textit{simple\_world}, and \textit{simple\_adversary} scenarios. \textit{Simple\_tag} is a predator-prey task where 3 slower predators coordinate to capture a faster prey. \textit{Simple\_world} involves 4 slower agents to catch 2 faster adversaries that desire to eat food. \textit{Simple\_adversary} involves 2 cooperating agents and 1 adversary where agents need to reach a specified target landmark of two landmarks while the adversary is unaware of the target. We compare DEMAR with several baselines including QMIX \cite{rashid_qmix_2018}, S-QMIX \cite{pan_regularized_2021}, SM2-QMIX \cite{gan_stabilizing_2021}, TD3-QMIX \cite{ackermann_reducing_2019} and WCU-QMIX \cite{sarkar_weighted_2021}. All methods are implemented in the pymarl framework \cite{samvelyan_starcraft_2019} and use one-step return. The details of each baseline are provided in Appendix C. We follow the training and evaluation settings of Pan et al. \cite{pan_regularized_2021}. For the hyperparameter tuning, we conduct a grid search for each baseline. To reduce the load of the hyperparameter tuning for DEMAR, we adopt a heuristic sequential searching to tune the hyperparameters of DEMAR, which is provided in Appendix E with the hyperparameter settings in MPE. For each method, we run 5 independent trials with different random seeds and the resulting plots include the mean performance as well as the shaded 95\% confidence interval in Figure~\ref{tag}-\ref{adversary}. Besides, we plot the estimated $Q_{tot}$ in the log scale to show the overestimation status of each method in Figure~\ref{tagoverestimation}-\ref{adversaryoverestimation}.

As shown in Figure~\ref{figure:mpe}, DEMAR prevents overestimation and stabilizes the learning on all three MPE tasks. QMIX fails on all tasks as it gets the most severe overestimation as shown in Figure~\ref{tagoverestimation}-\ref{adversaryoverestimation}. Meanwhile, although WCU-QMIX, S-QMIX, and SM2-QMIX tackle the \textit{simple\_tag} task, they fail on the two other tasks. While TD3-QMIX controls the overestimation in \textit{simple\_tag} and \textit{simple\_world}, it cannot tackle \textit{simple\_adversary} where the overestimation on this task is the most severe. Overall, DEMAR is the only method that successfully stabilizes the learning on all tasks in MPE.

\subsection{Experiments in Noisy SMAC}
\label{sec:noisysmac}

Next, we conduct the experiments on a noisy version of the commonly used MARL benchmark SMAC \cite{samvelyan_starcraft_2019}. In this noisy SMAC environment, the interference signal noises are added to the sensors of each agent's observation and the global state (details in Appendix D). The noise increases the variances of the individual Q-value and global Q-value, and thus raises the overestimation problem for multiagent Q-learning algorithms, making the noisy SMAC environment a good testbed for MARL algorithms designed to address the multiagent overestimation problem. We perform the experiments on 4 tasks including \textit{5m\_vs\_6m}, \textit{2s3z}, \textit{3s5z}, and \textit{10m\_vs\_11m}. In \textit{5m\_vs\_6m}, there are 5 allied marines against 6 marine enemies. In the map \textit{2s3z}, both sides have 2 Stalkers and 3 Zealots. For map \textit{3s5z}, both sides have 3 Stalkers and 5 Zealots. In \textit{10m\_vs\_11m}, there are 10 allied marines against 11 marine enemies. We train multiple agents to control allies respectively while a built-in handcrafted AI controls the enemies. Training and evaluation schedules such as the testing setting and training hyperparameters are kept unchanged \cite{samvelyan_starcraft_2019}. The version of StarCraft II is 4.6.2.

We compare DEMAR with QMIX, SM2-QMIX, TD3-QMIX, WCU-QMIX, and Sub-Avg-QMIX \cite{wu_sub-avg_2022}. Here we add the Sub-Avg-QMIX which performs well on standard SMAC tasks \cite{wu_sub-avg_2022}. We omit S-QMIX's results as it collapses while training. Similarly, we conduct the grid hyperparameter searching for baselines and use the sequential hyperparameter searching for DEMAR. The hyperparameter settings in SMAC are in Appendix E. Results are averaged over 6 independent training trials with different random seeds, and the resulting plots include the median performance as well as the shaded 25-75\% percentiles. Here we report the episode reward as the performance metric instead of the test winrate because the test winrate is almost zero for some baselines in most tasks. The results are shown in Figure~\ref{5m6m}-\ref{10m11m}. We also report the estimated $Q_{tot}$ in the log scale in Figure~\ref{5m6moverestimation}-\ref{10m11moverestimation}.

Figure~\ref{figure:smac} shows that, in the noisy SMAC, QMIX suffers from severe overestimation and cannot learn stably on all tasks. Meanwhile, although Sub-Avg-QMIX controls the overestimation and stabilizes learning on \textit{2s3z} and \textit{3s5z}, it fails on \textit{5m\_vs\_6m} and \textit{10m\_vs\_11m}. At the same time, SM2-QMIX performs well on \textit{5m\_vs\_6m} by successfully controlling overestimation on this task but fails on the other three tasks. Both TD3-QMIX and WCU-QMIX fail to control the overestimation on most tasks. While other baselines cannot consistently tackle all the tasks, DEMAR successfully reduces the overestimation to stabilize the learning in these four maps, which validates DEMAR's effectiveness. 

\begin{figure}[htbp]
\centering
\subfigure[simple\_tag ablation]{
\label{tagablation}
\includegraphics[width=0.231\textwidth]{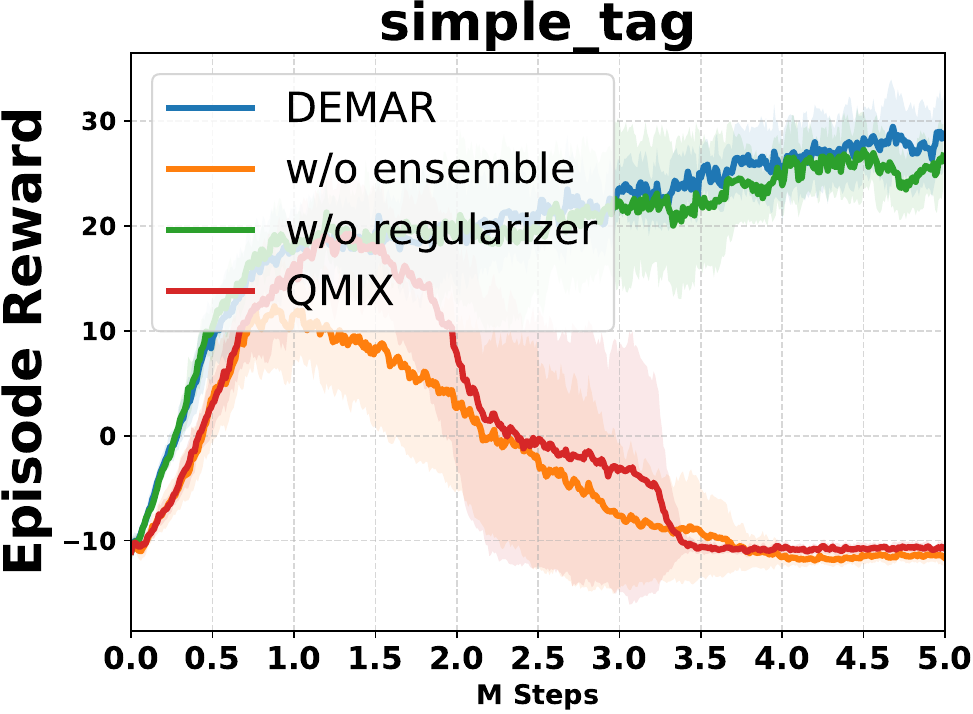}}
\subfigure[$Q_{tot}$ on simple\_tag]{
\label{tagablationoverestimation}
\includegraphics[width=0.231\textwidth]{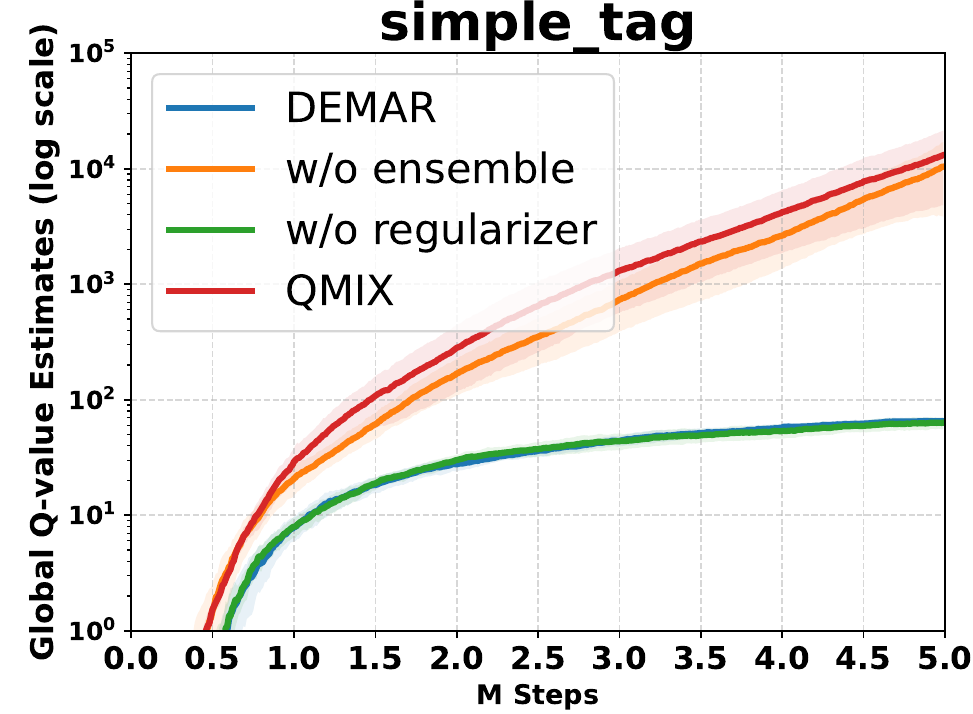}}
\subfigure[5m\_vs\_6m ablation]{
\label{5m6mablation}
\includegraphics[width=0.231\textwidth]{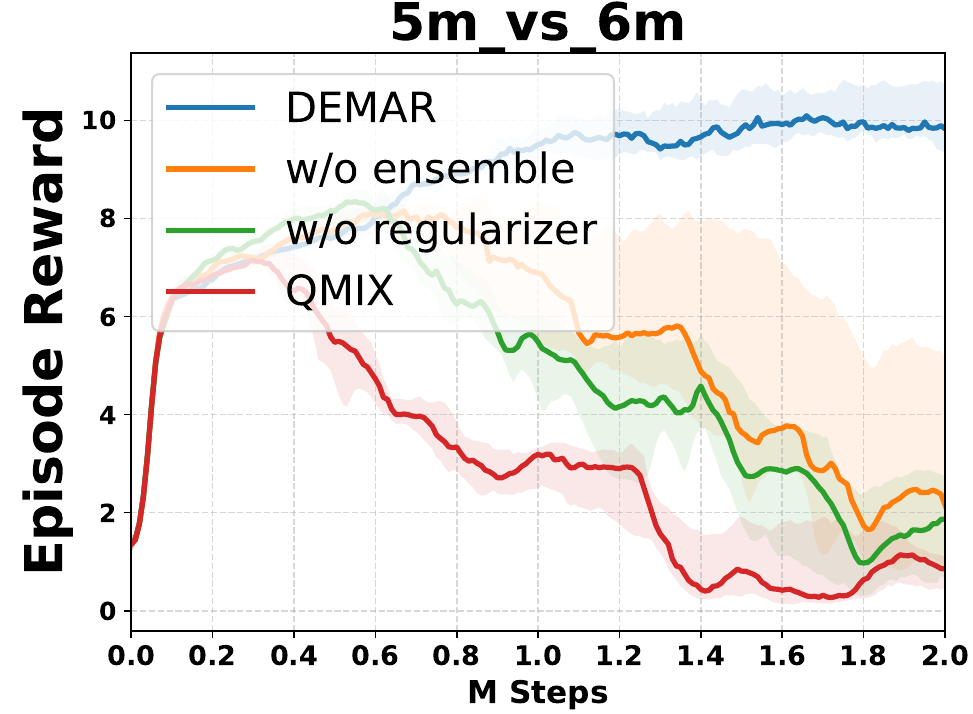}}
\subfigure[$Q_{tot}$ on 5m\_vs\_6m]{
\label{5m6mablationoverestimation}
\includegraphics[width=0.231\textwidth]{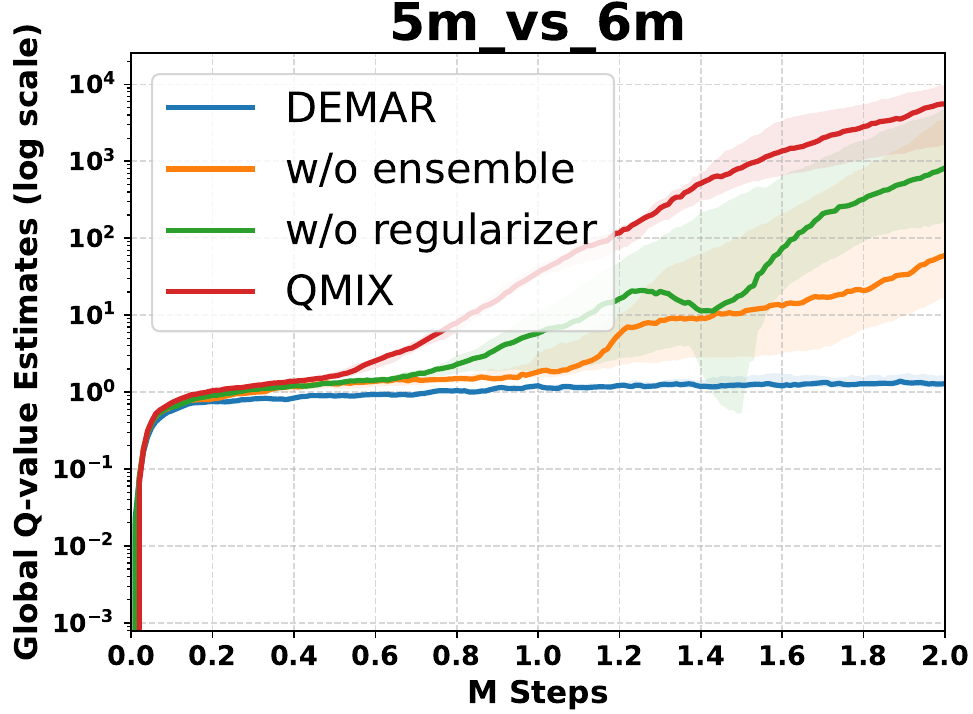}}
\caption{Results of the ablation study on \textit{simple\_tag} and \textit{5m\_vs\_6m}. The w/o ensemble indicates DEMAR without the dual ensembled Q-learning. The w/o regularizer represents DEMAR without the hypernet regularizer.}
\label{figure:ablation}
\end{figure}

\subsection{Ablation Study of DEMAR}

In this section,  we perform the ablation study to validate each component of DEMAR. We compare DEMAR, DEMAR without the dual ensembled Q-learning (w/o ensemble), and DEMAR without the hypernet regularizer (w/o regularizer). When DEMAR is without both techniques, it degenerates to vanilla QMIX. Figure~\ref{figure:ablation} shows the ablation results on \textit{simple\_tag} from MPE and \textit{5m\_vs\_6m} from SMAC. As we can see, both techniques of DEMAR are essential to address the multiagent overestimation problem. Especially, in \textit{5m\_vs\_6m}, neither the dual ensembled Q-learning nor the hypernet regularizer addresses the overestimation alone. Only the combined techniques, which jointly control the overestimation terms, can avoid overestimation and successfully stabilize learning. We also conduct the ablation study with different numbers of ensemble networks and different coefficients of the hypernet regularization. The additional experimental results are provided in Appendix F.

\subsection{Experiment Study of Overestimation Terms}
\label{sec:emp_ana_cause}

\begin{figure}[htbp]
\centering
\subfigure[$Q_{tot}$ on simple\_adversary]{
\label{adversaryqtot}
\includegraphics[width=0.235\textwidth]{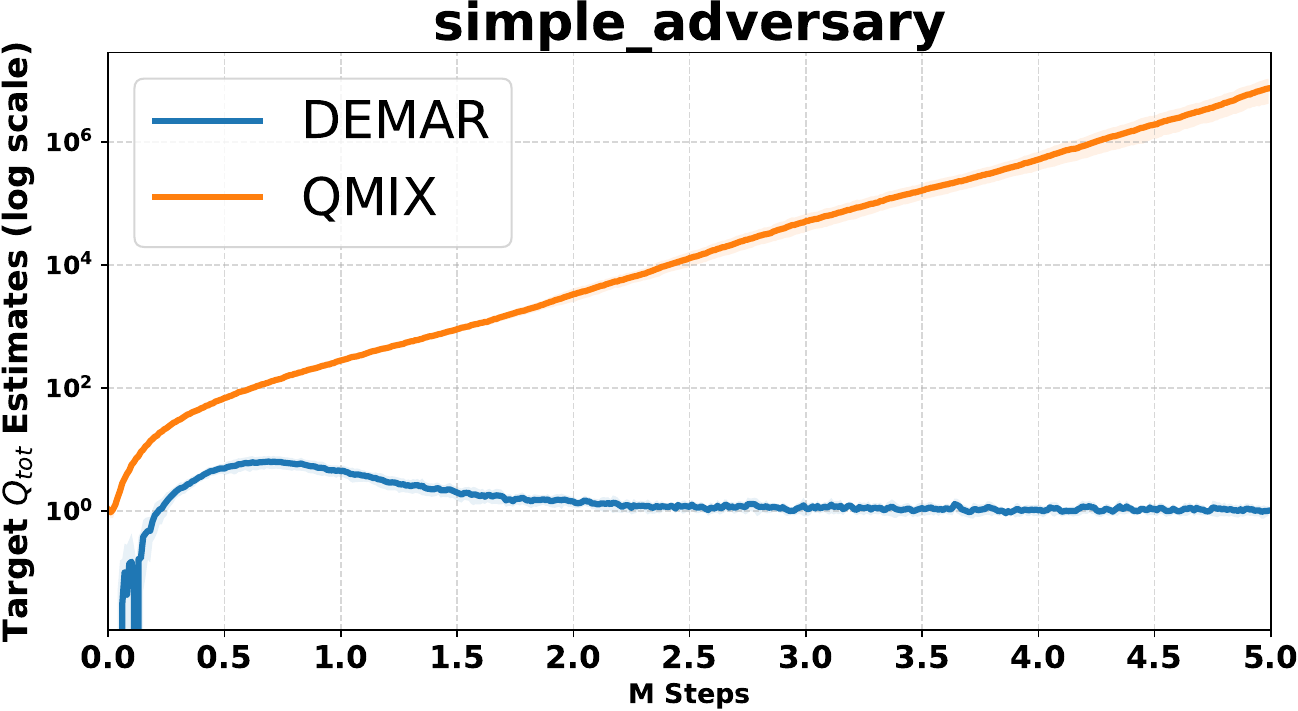}}
\subfigure[$Q_{tot}$ on 5m\_vs\_6m]{
\label{5m6mqtot}
\includegraphics[width=0.218\textwidth]{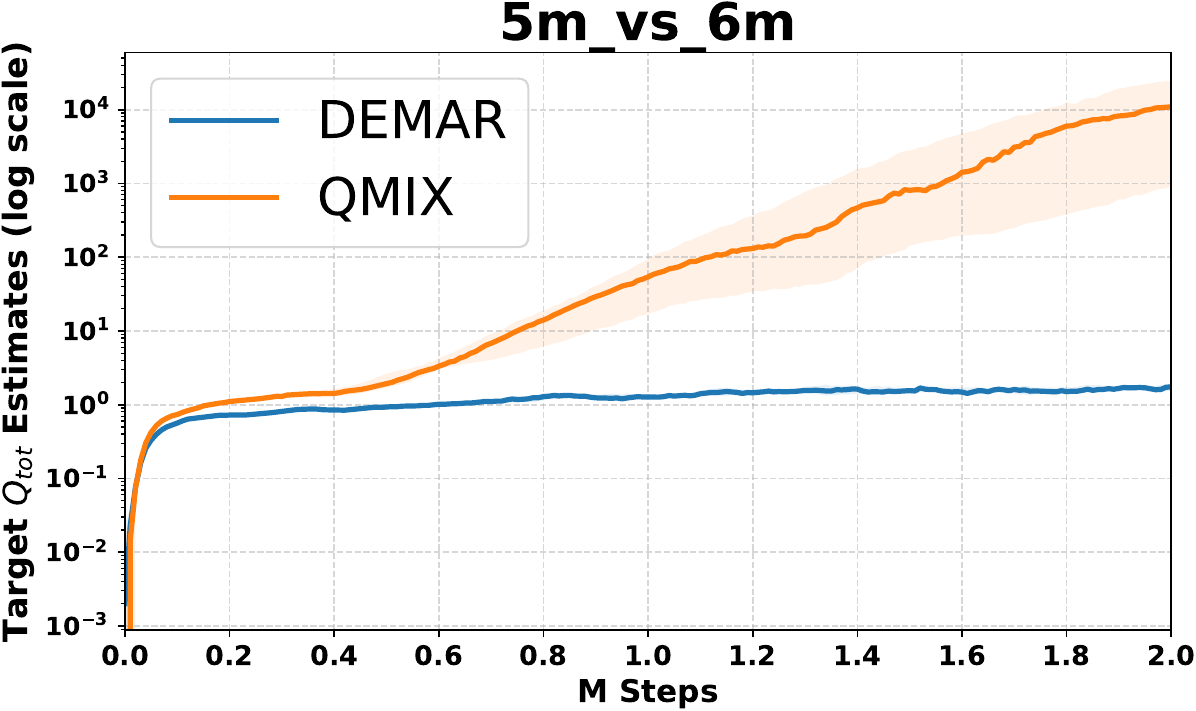}}
\subfigure[$Q_{i}$ on simple\_adversary]{
\label{adversaryagentutility}
\includegraphics[width=0.235\textwidth]{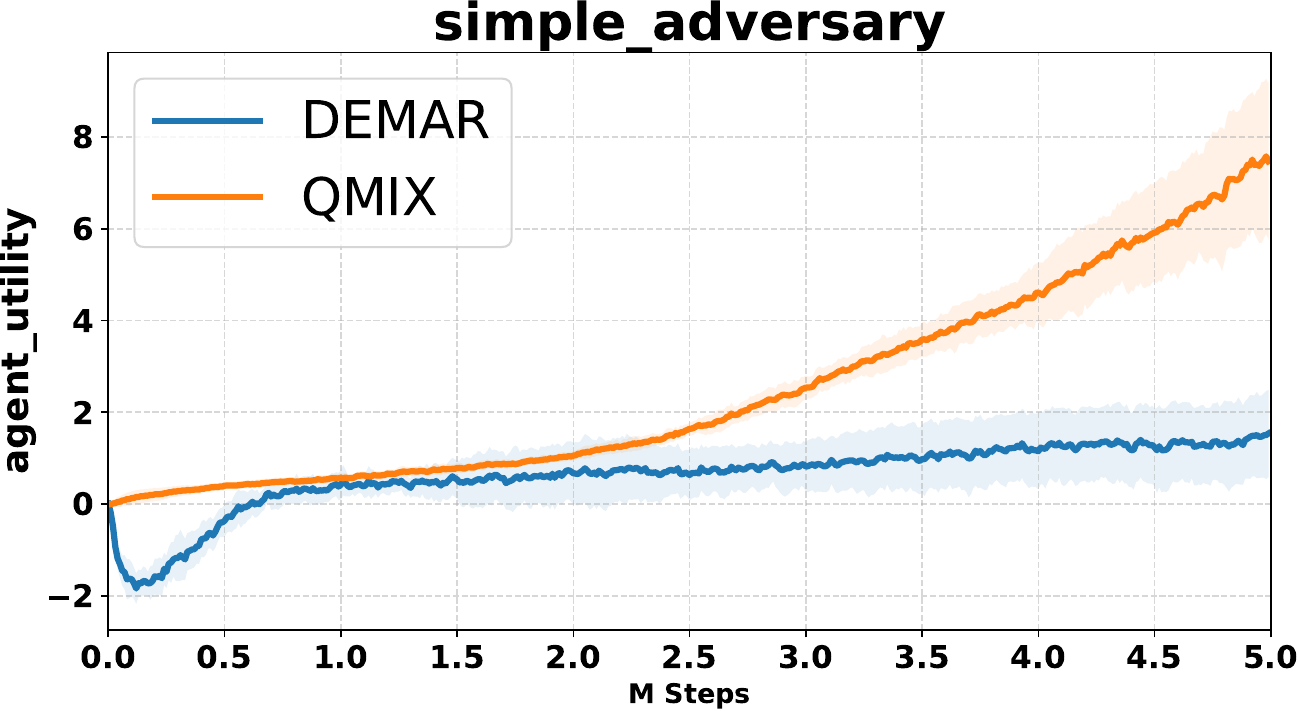}}
\subfigure[$Q_{i}$ on 5m\_vs\_6m]{
\label{5m6magentutility}
\includegraphics[width=0.215\textwidth]{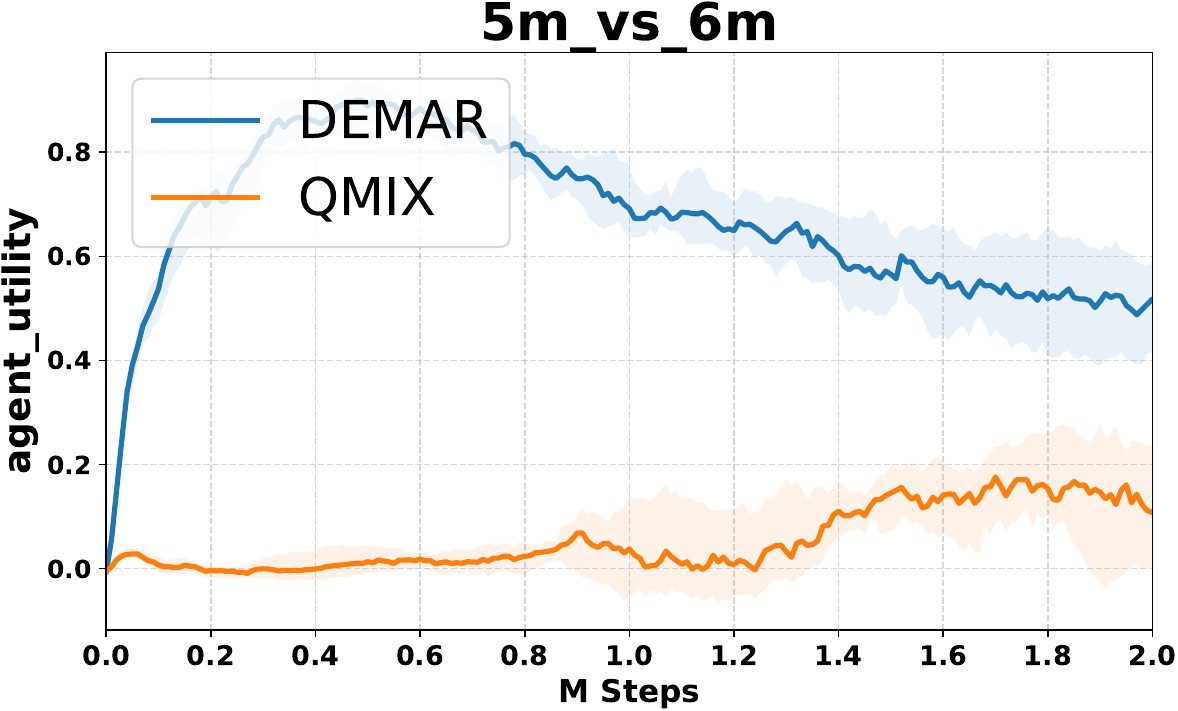}}
\subfigure[$\frac{\partial Q_{tot}}{\partial Q_{i}}$ on simple\_adversary]{
\label{adversarygradient}
\includegraphics[width=0.235\textwidth]{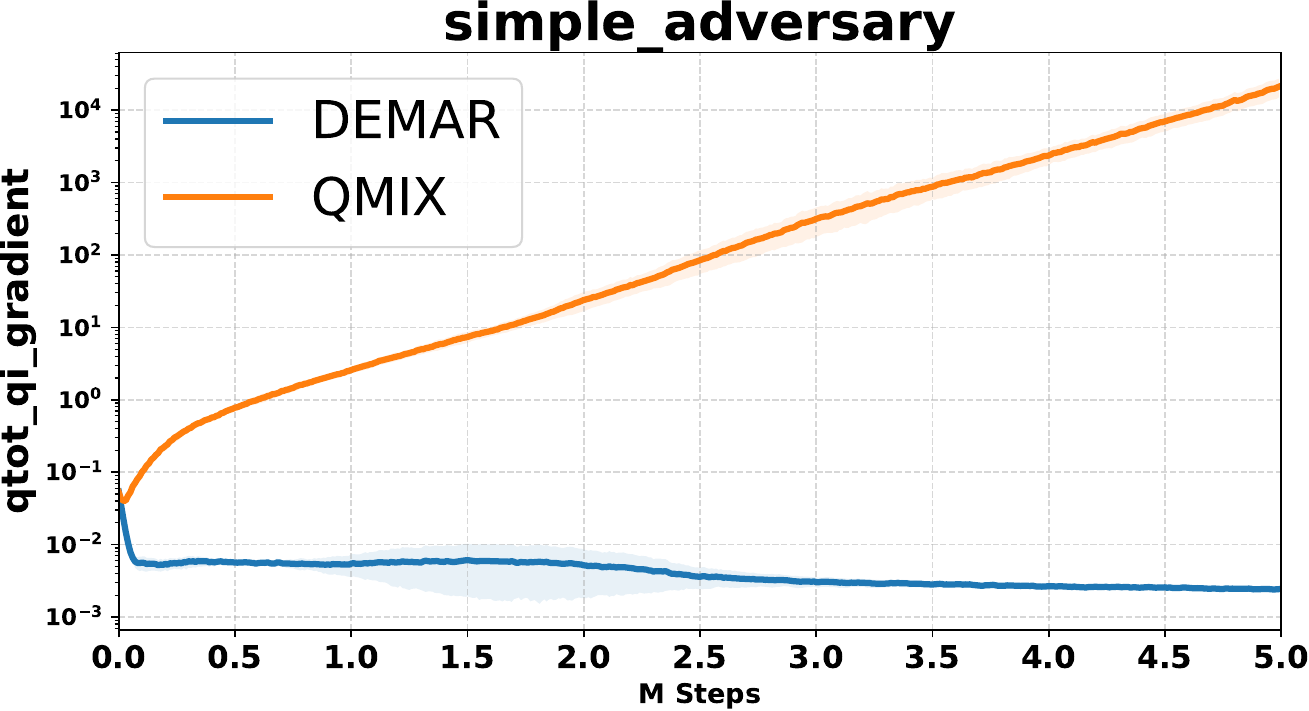}}
\subfigure[$\frac{\partial Q_{tot}}{\partial Q_{i}}$ on 5m\_vs\_6m]{
\label{5m6mgradient}
\includegraphics[width=0.216\textwidth]{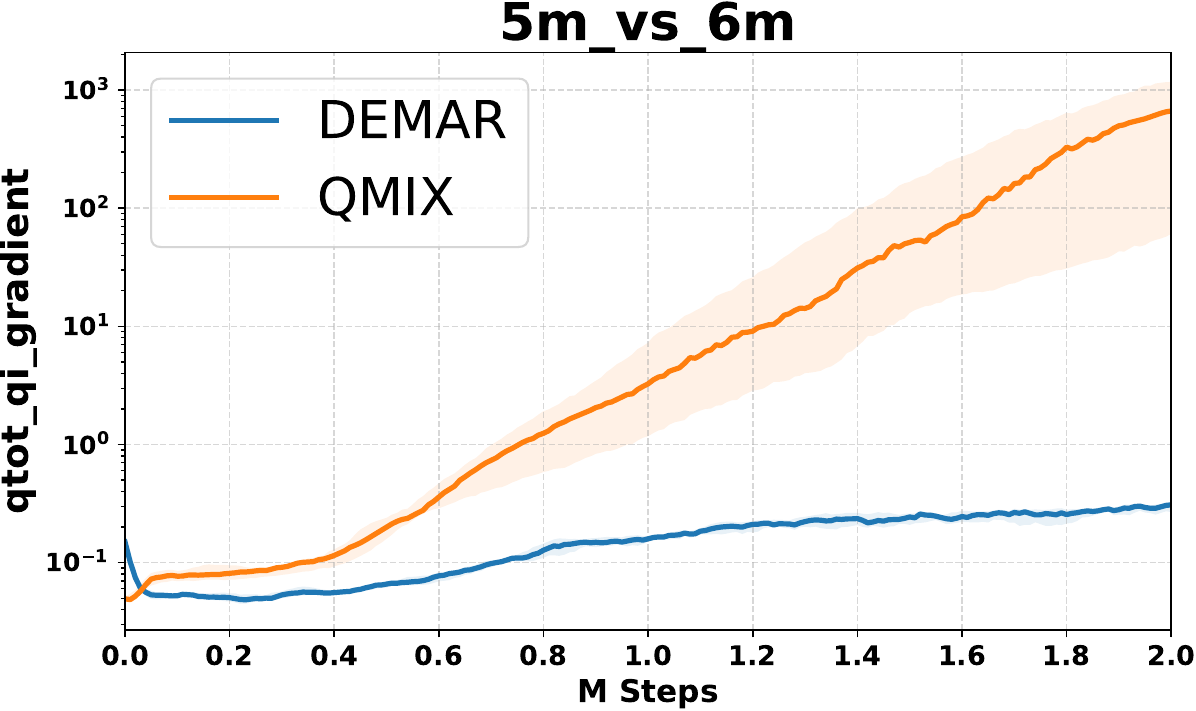}}
\caption{Results of analyzed overestimation terms including $Q_{tot}$, $Q_{i}$, and $\frac{\partial Q_{tot}}{\partial Q_{i}}$ on both \textit{simple\_adversary} and \textit{5m\_vs\_6m}.}
\label{figure:rootcauses}
\end{figure}

As DEMAR is designed to control each analyzed overestimation term, we experimentally examine $Q_{i}$ and $Q_{tot}$ in the target Q-value estimation as well as $\frac{\partial Q_{tot}}{\partial Q_{i}}$ in the online Q-network optimization to see whether DEMAR works as expected. Here we use the \textit{simple\_adversary} task because its overestimation is the most severe among all tasks. The values of each analyzed overestimation term are plotted in Figure~\ref{adversaryqtot}-\ref{adversarygradient}. We see that each overestimation term in QMIX has higher values than DEMAR. Especially, $\frac{\partial Q_{tot}}{\partial Q_{i}}$ contributes much to the overestimation of QMIX as the value of $\frac{\partial Q_{tot}}{\partial Q_{i}}$ in QMIX is much larger than in DEMAR by orders of magnitude. Overall, the experimental results here correspond to our analysis in Theorem~\ref{severetheorem}. By controlling each analyzed overestimation term, DEMAR successfully prevents severe multiagent overestimation.

\subsection{Comparing True and Estimated Q-values}

\begin{figure}[htbp]
\centering
\subfigure[$Q_{tot}$ on 5m\_vs\_6m]{
\includegraphics[width=0.220\textwidth]{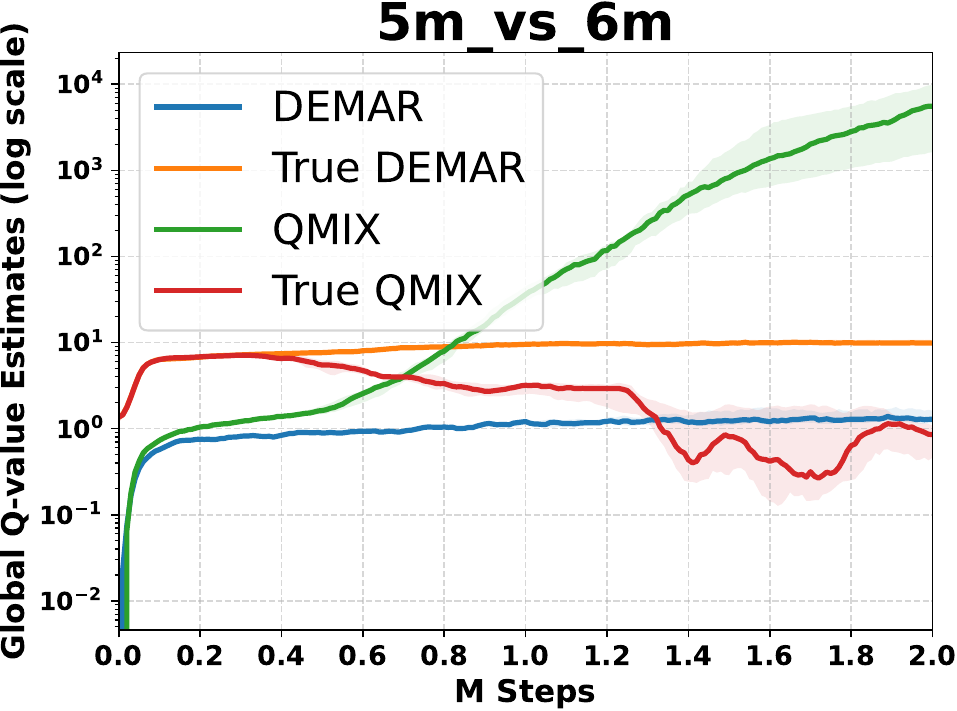}}
\subfigure[$Q_{tot}$ on 2s3z]{
\includegraphics[width=0.220\textwidth]{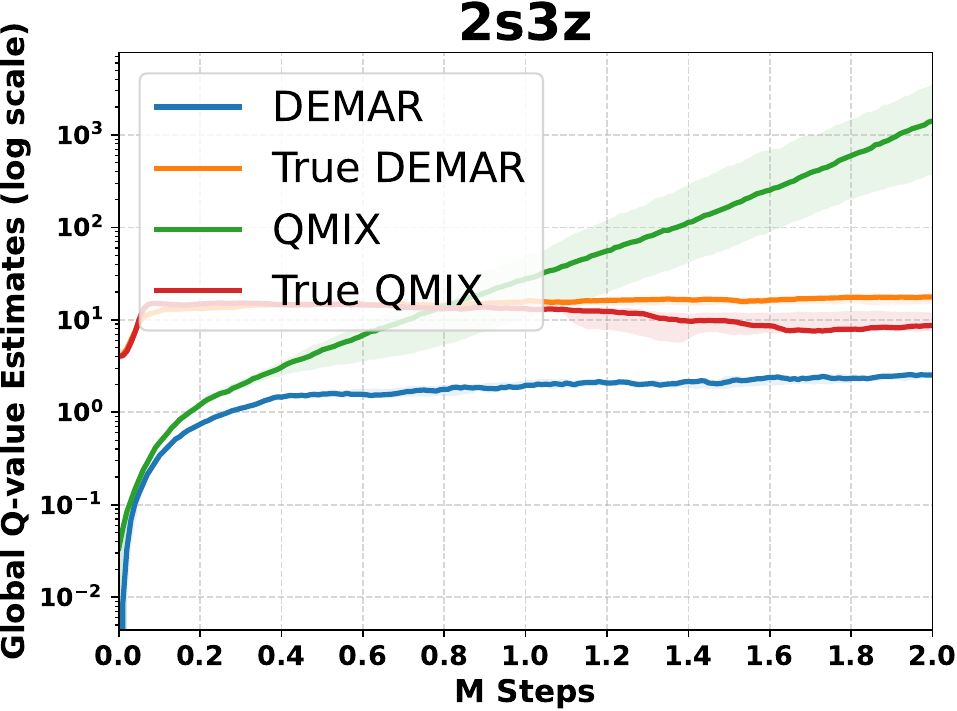}}
\subfigure[$Q_{tot}$ on 3s5z]{
\includegraphics[width=0.220\textwidth]{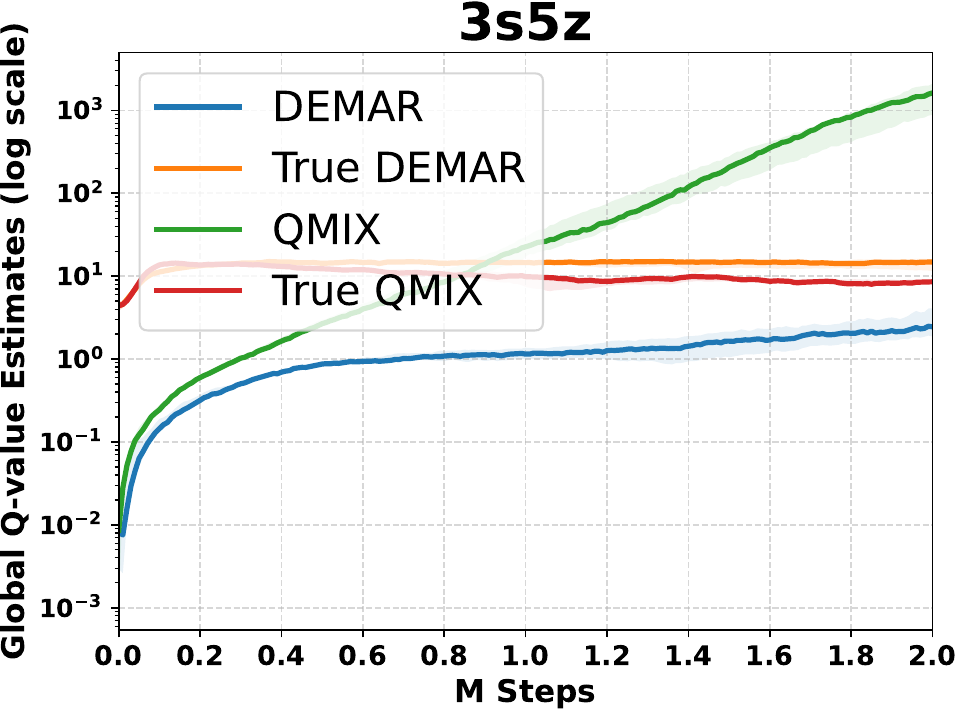}}
\subfigure[$Q_{tot}$ on 10m\_vs\_11m]{
\includegraphics[width=0.220\textwidth]{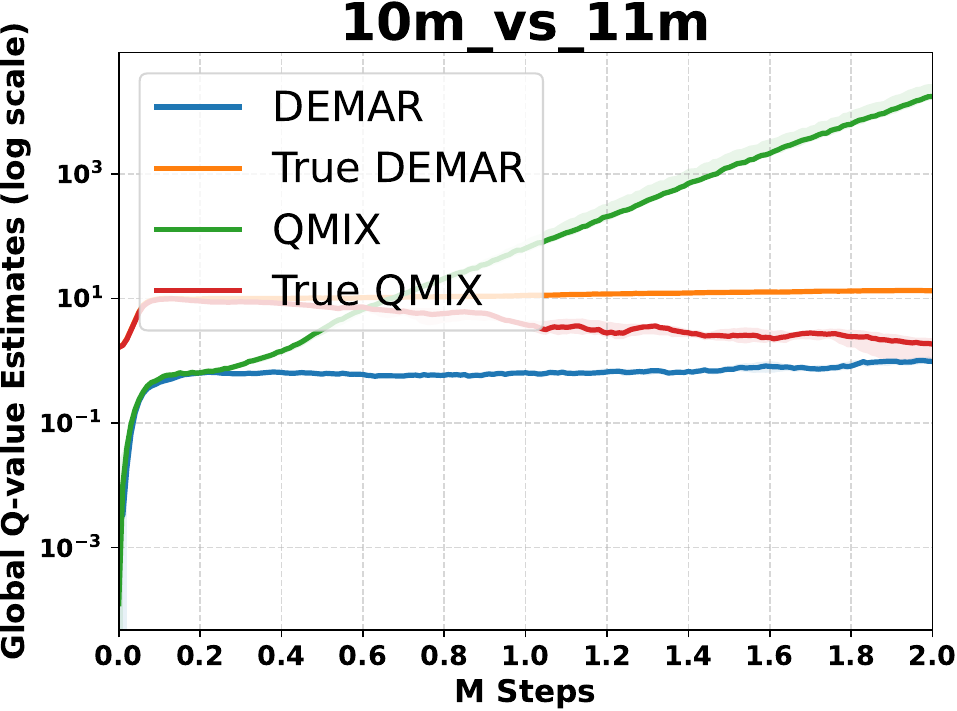}}
\caption{Comparison between True Q-values and Estimated Q-values. All Q-value curves are plotted in the log scale.}
\label{figure:trueQ}
\end{figure}

In this section, we compare the true Q-values and estimated Q-values of DEMAR as well as QMIX. We estimate true values by summing up the discounted returns of the following transitions starting from the sampled state. The results are shown in Figure~\ref{figure:trueQ}. We could see that the estimated Q-value by DEMAR is closer to its true Q-value than QMIX. At the same time, the estimated Q-value by QMIX increases rapidly and its gap to its true Q-value also becomes larger as training continues.

\subsection{Extending DEMAR to Other MARL Methods}

\begin{figure}[htbp]
\centering
\subfigure[5m\_vs\_6m]{
\label{5m6masnupdet}
\includegraphics[width=0.228\textwidth]{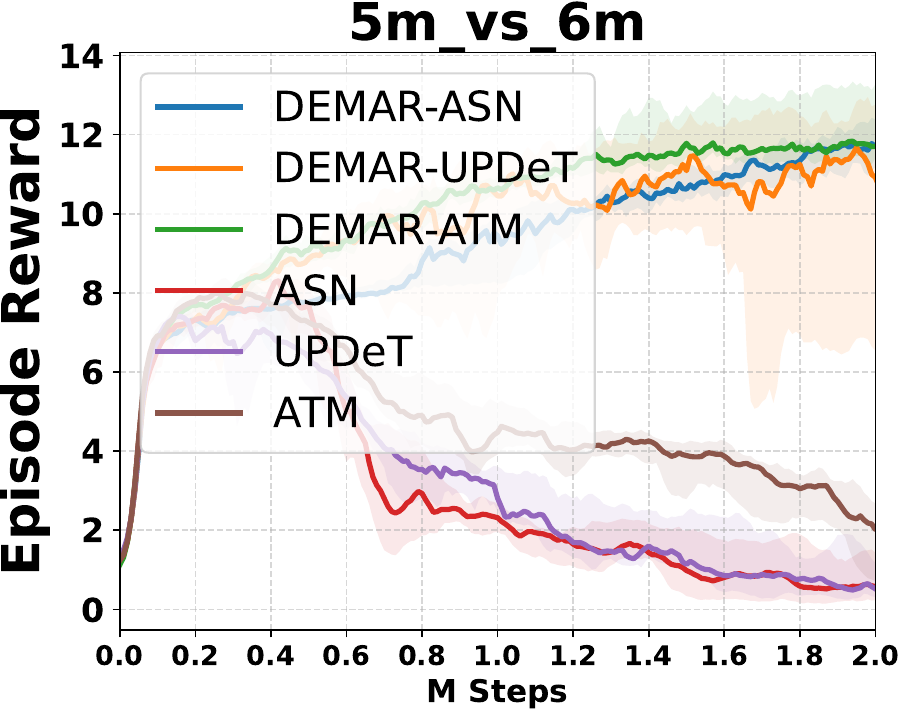}}
\subfigure[$Q_{tot}$ on 5m\_vs\_6m]{
\label{5m6masnupdetoverestimation}
\includegraphics[width=0.228\textwidth]{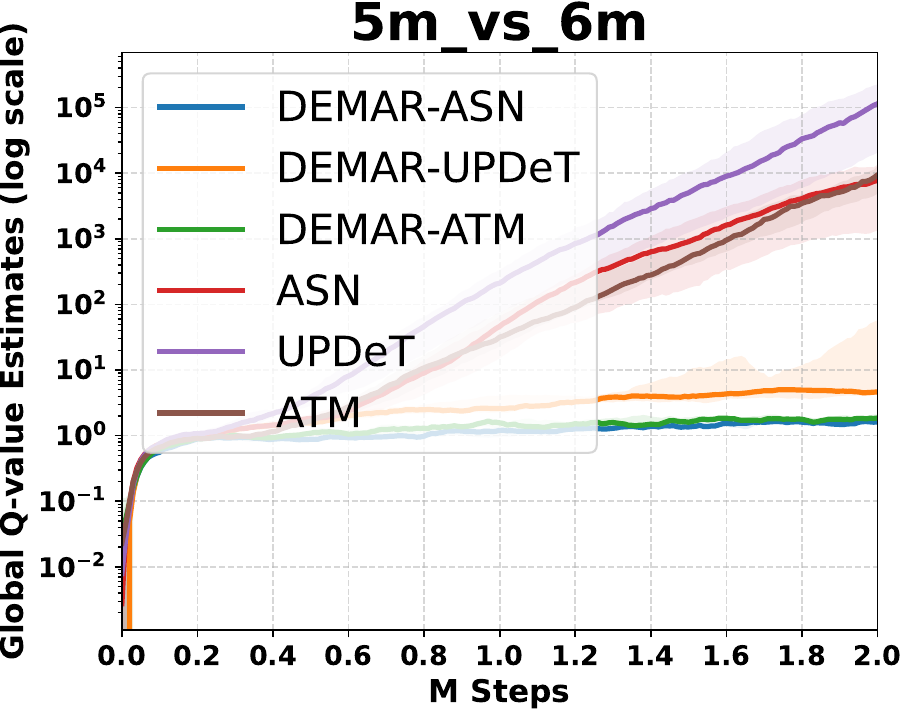}}
\subfigure[10m\_vs\_11m]{
\label{10m11masnupdet}
\includegraphics[width=0.228\textwidth]{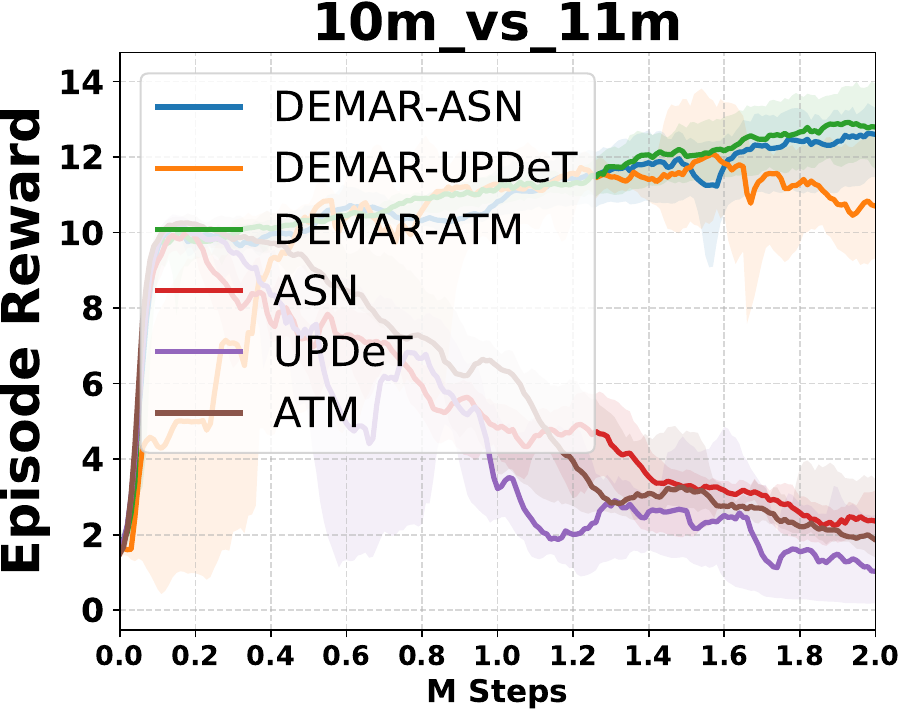}}
\subfigure[$Q_{tot}$ on 10m\_vs\_11m]{
\label{10m11masnupdetoverestimation}
\includegraphics[width=0.228\textwidth]{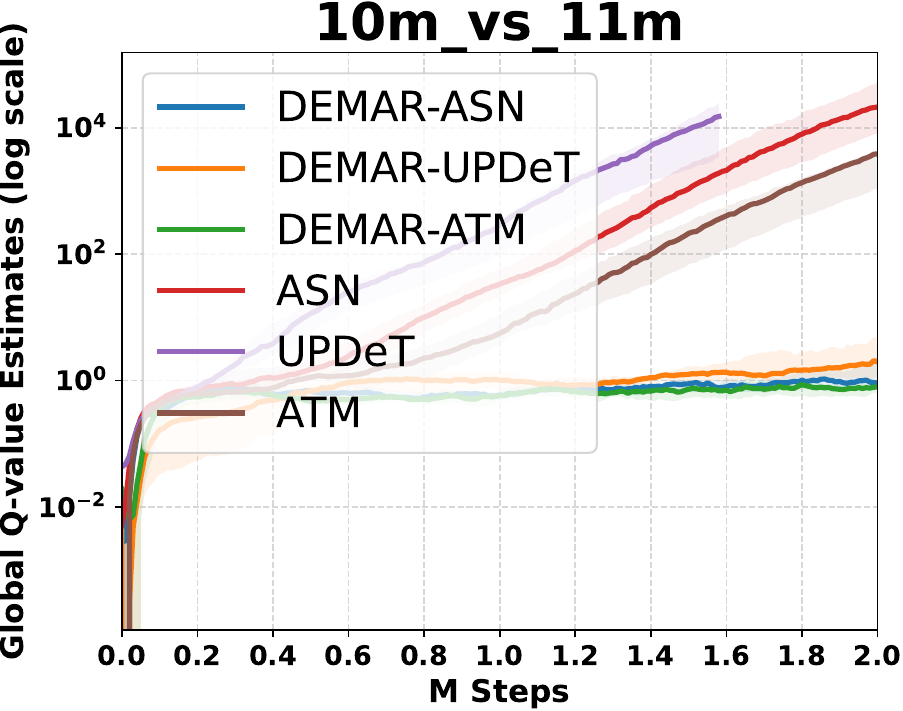}}
\caption{Results of UPDeT, ASN, and ATM with DEMAR.}
\label{figure:extend}
\end{figure}

To test the generality of DEMAR, we extend DEMAR to other advanced MARL algorithms such as ASN \cite{wang_action_2020}, UPDeT \cite{hu_updet_2021}, and ATM \cite{yang_transformer-based_2022}. ASN explicitly represents action semantics between agents and characterizes different actions' influence on other agents using neural networks to significantly improve the performance of MARL algorithms \cite{wang_action_2020}. UPDeT utilizes a transformer-based model to enable multiple tasks transferring in MARL through the transformer’s strong generalization abilities \cite{hu_updet_2021}. ATM proposes a transformer-based working memory mechanism to address partial observability in multiagent scenarios \cite{yang_transformer-based_2022}. The results of extending DEMAR into ASN, UPDeT, and ATM are shown in Figure~\ref{figure:extend}. As we can see, DEMAR helps ASN, UPDeT, and ATM avoid severe overestimation and stabilizes the learning process, which validates the generality of DEMAR. By the way, the global Q-value estimation curve of UPDeT on \textit{10m\_vs\_11m} is clipped in Figure~\ref{10m11masnupdetoverestimation}. During running UPDeT on \textit{10m\_vs\_11m}, the global Q-value estimation in one trial increased too much to become NaN (too large to represent) while learning. Then the mean value of global Q-value estimation of UPDeT on \textit{10m\_vs\_11m} also becomes NaN and is clipped when plotting.

\section{Conclusion}
In this study, we propose DEMAR to address the challenging multiagent overestimation problem. For the first time, we establish an iterative estimation-optimization analysis framework to systematically analyze the overestimation in multiagent value-mixing Q-learning. We found that the multiagent overestimation not only comes from the overestimation of target individual and global Q-values but also accumulates in the online Q-network's optimization. Motivated by this analysis finding, we propose DEMAR with dual ensembled Q-learning and hypernet regularizer to address these analyzed overestimation terms correspondingly. Extensive experiments in MPE and the noisy SMAC demonstrate that DEMAR successfully controls the multiagent overestimation.

For future work, on the one hand, there is a high potential to apply DEMAR to real-world multiagent scenarios where environmental noises are common and learning stability is a prerequisite. And it would be helpful to develop evaluation techniques for MARL in real-world applications. On the other hand, extending DEMAR to policy-based MARL algorithms is also promising. 



\begin{acks}
This work is supported by the National Natural Science Foundation of China (62376254, 32341018), a grant from the Research Grants Council of the Hong Kong Special Administrative Region, China (Project No. CUHK 14201321), and a grant from Hong Kong Innovation and Technology Fund (Project No. MHP/092/22).
\end{acks}



\bibliographystyle{ACM-Reference-Format} 
\balance
\bibliography{sample}


\clearpage
\appendix
\onecolumn

\section{Proof of Overestimation in Value-mixing Q-learning}
\label{appendix:maoverestimationproof}

Here we provide the proof for completeness. Gan et al. \cite{gan_stabilizing_2021} show that if the assumption that $l \leq \frac{\partial  Q_{tot}}{\partial Q_{i}} \leq L, i=1,2,...,N$ where $l \geq 0$ and $L > 0$ satisfies, multiagent Q-learning algorithms with the monotonic value-mixing global Q-network such as VDN \cite{sunehag_value-decomposition_2018}, QMIX \cite{rashid_qmix_2018}, and Qatten \cite{yang_qatten_2020} obtain the overestimation
\begin{equation}
        LN\mathbb{E}[Z_{i}^{s}] \geq \mathbb{E}[r + \gamma \max_{\mathbf{a}'} Q_{tot}(s',\mathbf{Q}(s', \mathbf{a}_{i}')) - (r + \gamma \max_{\mathbf{a}'} Q_{tot}(s',\mathbf{Q^{*}}(s', \mathbf{a}_{i}')))] \geq lN\mathbb{E}[Z_{i}^{s}],
\end{equation}
where $\mathbb{E}[Z_{i}^{s}] = \mathbb{E}[\max_{a'_{i}} Q_{i}(s',a'_{i}) - \max_{a'_{i}} Q_{i}^{*}(s',a'_{i})]$ and $\mathbf{Q^{*}}$ are the optimal target individual Q-values.

\begin{proof}
\begin{equation}
\label{eq:maoverestimationproof}
\begin{split}
        & \mathbb{E}[r + \gamma \max_{\mathbf{a}'} Q_{tot}(s',\mathbf{Q}(s', \mathbf{a}'_{i})) - (r + \gamma \max_{\mathbf{a}'} Q_{tot}(s',\mathbf{Q^{*}}(s', \mathbf{a}'_{i})))] \\
        & = \gamma \mathbb{E}[Q_{tot}(s',\max_{\mathbf{a}'} \mathbf{Q}(s', \mathbf{a}'_{i})) - Q_{tot}(s',\max_{\mathbf{a}'} \mathbf{Q^{*}}(s', \mathbf{a}'_{i}))] \\
        & = \gamma \mathbb{E}[Q_{tot}(s',\max_{a'_{1}}Q_{1}(s', a'_{1}),...,\max_{a'_{N}}Q_{N}(s', a'_{N})) - Q_{tot}(s',\max_{a'_{1}} Q_{1}^{*}(s',a'_{1}),...,\max_{a'_{N}} Q_{N}^{*}(s', a'_{N}))] \\
        & \geq \gamma \mathbb{E}[\sum_{i}^{N} l (\max_{a'_{i}}Q_{i}(s',a'_{i}) - \max_{a'_{i}} Q_{i}^{*}(s',a'_{i}))] \\
        & = \gamma lN\mathbb{E}[\max_{a'_{i}} Q_{i}(s',a'_{i}) - \max_{a'_{i}} Q_{i}^{*}(s',a'_{i})] \\
        & = lN\mathbb{E}[Z_{i}^{s}],
\end{split}
\end{equation}
where the estimated $Q_{i}$ is assumed with an independent noise uniformly distributed in $[-\epsilon, \epsilon]$ on each action $a_{i}$ given $s$. Similarly, we can also get $\mathbb{E}[r + \gamma \max_{\mathbf{a}'} Q_{tot}(s',\mathbf{Q}(s', \mathbf{a}'_{i})) - (r + \gamma \max_{\mathbf{a}'} Q_{tot}(s',\mathbf{Q^{*}}(s', \mathbf{a}'_{i})))] \leq LN\mathbb{E}[Z_{i}^{s}]$.
\end{proof}

\section{Proof of Hypernet Regularizer}
\label{appendix:regularizerproof}
Multiagent value-mixing algorithms use the outputted variables from hypernetworks as the weights and biases for the mixing network layers to transform the individual $Q_{i}$s into the global $Q_{tot}$. For the most representative QMIX \cite{rashid_qmix_2018}, the global Q-value is calculated as follows
\begin{equation}
    Q_{tot}=f_{mix}(s,Q_{1},...,Q_{N})=elu(\mathbf{Q}_{i}^{1 \times N} \mathbf{W}_{f,1}^{N \times L_h} + \mathbf{B}_{f,1}^{1 \times L_h}) \mathbf{W}_{f,2}^{L_h \times 1} + b_{f,2}^{1 \times 1},
\end{equation}
where $\mathbf{W}_{f,1}^{N \times L_h}$ and $\mathbf{W}_{f,2}^{L_h \times 1}$ are weights while $\mathbf{B}_{f,1}^{1 \times L_h}$ and $b_{f,2}^{1 \times 1}$ are biases generated from the corresponding hypernetworks. $L_h$ is the hidden unit number. First, we show the relation between $\frac{\partial Q_{tot}}{\partial Q_{i}}$ and the weights and biases outputted from hypernetworks.
\begin{equation}
\label{eq:qmixgrad}
\begin{aligned}
    \frac{\partial Q_{tot}}{\partial Q_{i}} & = \frac{\partial elu(\mathbf{Q}_{i}^{1 \times N} \mathbf{W}_{f,1}^{N \times L_h} + \mathbf{B}_{f,1}^{1 \times L_h}) \mathbf{W}_{f,2}^{L_h \times 1} + b_{f,2}^{1 \times 1}}{\partial Q_{i}} \\
    & = \frac{\partial elu(Q_{i}^{1 \times 1} \mathbf{W}_{f,1,i}^{1 \times L_h} + B_{f,1}^{1 \times L_h})\mathbf{W}_{f,2}^{L_h \times 1}}{\partial Q_{i}} \\
    & = \sum_{\substack{l_h=1\\Q_{i}w_{1,i,l_h}+b_{1,l_h} \geq 0}}^{L_h} w_{1,i,l_h}w_{2,l_h} + \sum_{\substack{l_h=1\\Q_{i}w_{1,i,l_h}+b_{1,l_h} < 0}}^{L_h} \alpha_{elu} w_{1,i,l_h}w_{2,l_h}e^{Q_{i}w_{1,i,l_h}+b_{1,l_h}} \\
    & \leq \sum_{\substack{l_h=1\\Q_{i}w_{1,i,l_h}+b_{1,l_h} \geq 0}}^{L_h} w_{1,i,l_h}w_{2,l_h} + \sum_{\substack{l_h=1\\Q_{i}w_{1,i,l_h}+b_{1,l_h} < 0}}^{L_h} \alpha_{elu} w_{1,i,l_h}w_{2,l_h},
\end{aligned}
\end{equation}
where $w_{1,i,l_h} \in \mathbf{W}_{f,1,i}^{1 \times L_h} \geq 0$, $w_{2,l_h} \in \mathbf{W}_{f,2}^{L_h \times 1} \geq 0$, $b_{1,l_h} \in \mathbf{B}_{f,1}^{1 \times L_h}$, $b_{2}=b_{f,2}$, and $elu$ is the Exponential Linear Unit activation function, and $\alpha_{elu} > 0$ is a scalar \cite{clevert_fast_2016}. Therefore, if we use
\begin{equation}
    L_{reg}=\sum{|\mathbf{W}_{f}|} + \sum{|\mathbf{B}_{f}|} =\sum{w_{1}} + \sum{w_{2}} + \sum{|b_{1}|} + |b_{2}|
\end{equation}
as the regularization term in the loss function, we can constrain the term $\frac{\partial Q_{tot}}{\partial Q_{i}}$. 

Furthermore,
\begin{equation}
    \frac{\partial}{\partial Q_{i}} (\frac{\partial Q_{tot}}{\partial Q_{i}}) = \sum_{\substack{l_h=1\\Q_{i}w_{1,i,l_h}+b_{1,l_h} < 0}}^{L_h} \alpha_{elu} w_{1,i,l_h}^{2} w_{2,l_h}e^{Q_{i}w_{1,i,l_h}+b_{1,l_h}} \geq 0.
\end{equation}
Therefore, in QMIX, $\frac{\partial Q_{tot}}{\partial Q_{i}}$ increases with $Q_{i}$ (strictly increases when $Q_{i}w_{1, i, l_h}+b_{1,l_h} < 0$ for some $l_h$) and the overestimation can be accumulated, which is also observed from the experiments \cite{pan_regularized_2021}. With the hypernet regularizer, we could constrain $\frac{\partial Q_{tot}}{\partial Q_{i}}$ to prevent the accumulation of overestimation. 

For the multiagent value-mixing algorithms in the linear form \cite{yang_qatten_2020}, the global Q-value is calculated as
\begin{equation}
    Q_{tot}= \sum_{i=1}^{N} w_{i}Q_{i} + b,
\end{equation}
where $w_{i} \geq 0$. The following simply holds as
\begin{equation}
    \frac{\partial Q_{tot}}{\partial Q_{i}}=w_{i}.
\end{equation}
As $w_{i} \geq 0$, the hypernet regularizer becomes
\begin{equation}
    L_{reg}=\sum{|\mathbf{W}_{f}|} + \sum{|\mathbf{B}_{f}|} =\sum{w} + |b|.
\end{equation}
Therefore, $L_{reg}$ in the linear form $f_{mix}$ also regularizes $\frac{\partial Q_{tot}}{\partial Q_{i}}$. Here we also constrain the biases. During the optimization of the Q-value, the overestimation will be accumulated either by weights or biases. If we do not regularize the biases, the biases will increase to accumulate the overestimation as the global Q-network fits the overestimated global Q-target. 

\section{The Implementation Details of Baselines}
\label{appendix:baselines}

\subsection{S-QMIX}
S-QMIX uses the softmax Bellman operator to compute the estimation of the target global Q-value
\begin{equation}
\label{eq:sqmix}
    softmax_{\beta, \mathbf{U}}(Q_{tot}(s,\cdot)) = \sum_{\mathbf{u} \in \mathbf{U}} \frac{e^{\beta Q_{tot}(s,\mathbf{u})}}{\sum_{\mathbf{u}' \in \mathbf{U}}e^{\beta Q_{tot}(s,\mathbf{u}')}} Q_{tot}(s,\mathbf{u}),
\end{equation}
where $\beta \geq 0$ is the inverse temperature parameter. However, the computation of Eq.~(\ref{eq:sqmix}) in the multiagent setting can be computationally intractable as the size of the joint action space grows exponentially with the number of agents. Therefore, Pan et al. \cite{pan_regularized_2021} use an alternative joint action set $\mathbf{\hat{U}}$ to replace the joint action space $\mathbf{U}$. First, the maximal joint action $\mathbf{\hat{u}}$ is obtained by $\mathbf{\hat{u}}=\arg\max_{\mathbf{u}} Q(s, \mathbf{u})$. Next, for each agent $i$, $N_{\mathbf{u}}$ joint actions are considered by changing only agent $i$’s action while keeping the other agents’ actions $\mathbf{u}_{-i}$ fixed and the resulting action set of agent $i$ is $U_{i}=\{(u_{i}, \mathbf{\hat{u}}_{-i})|u_{i} \in U\}$. Finally, the joint action subspace $\mathbf{\hat{U}}=U_{i} \cup \cdot \cdot \cdot \cup U_{N}$ is obtained and used to calculate the softmax version of the global Q-value.

\subsection{SM2-QMIX}
SM2-QMIX uses the soft Mellowmax operator to compute the estimation of the target individual Q-value and thus avoids the explosion problem of the joint action space in S-QMIX.
\begin{equation}
    sm_{\omega} Q_{i}(s, \cdot) = \frac{1}{\omega} \log[\sum_{a \in A} \frac{e^{\alpha Q_{i}(s,a)}}{\sum_{a' \in A}e^{\alpha Q_{i}(s,a')}} e^{\omega Q_{i}(s,a)}],
\end{equation}
where $\omega > 0$ and $\alpha \in \mathbb{R}$, which can be viewed as a particular instantiation of the weighted quasi-arithmetic mean \cite{beliakov_practical_2015}.

\subsection{TD3-QMIX}
TD3-QMIX takes the minimum between the two critics' estimations to calculate the target global Q-value.
\begin{equation}
\begin{split}
    y_{tot} = r + \gamma & \min_{h \in \{1,2\}} Q_{tot}^{{\bar{\phi}_{h}}}(s',Q_{1}(o_{1}',a_{1}'),...,Q_{N}(o_{N}',a_{N}')), \\ & Q_{i}(o_{i}',a_{i}') = \max_{a_{i}'} Q_{i}^{\bar{\theta}_{i}}(o_{i}',a_{i}').
\end{split}
\end{equation}

\subsection{WCU-QMIX}
WCU-QMIX proposes a weighted critic updating scheme of TD3. It updates the critic networks with the loss function that is calculated using the weighted Q-values obtained from the two critic networks. The target is calculated using the rewards and the minimum of the target Q-values.
\begin{equation}
\begin{split}
    L(\phi_{h},\theta_{1},...,\theta_{N})=\frac{1}{|B|}\sum_{b}(y_{b} - (w Q_{tot}^{\phi_{h}} & (s_b,Q_{1},...,Q_{N}) + (1 - w)Q_{tot}^{\phi_{p}}(s_b,Q_{1},...,Q_{N}))|_{p \neq h}), \\
    & Q_{i} = \max_{a_{i}} Q_{i}^{\theta_{i}}(o_{i,b},a_{i}).
\end{split}
\end{equation}

\subsection{Sub-Avg-QMIX}
The Sub-Avg-QMIX keeps multiple target networks to maintain various action values of different periods and discards the larger action values to eliminate the excessive overestimation error. Thereby, Sub-Avg-QMIX gets an overall lower maximum action value and then obtains an overall lower update target. Specifically, Sub-Avg-QMIX discards the action values above the average. Here we apply the Sub-Avg operator in the mixing network as it shows better performance compared with the version of applying the Sub-Avg operator in the agent network \cite{wu_sub-avg_2022}.
\begin{equation}
    y_{tot} = r + \gamma \max_{\mathbf{a}'} (\frac{\sum_{k=1}^{K} c_{k} Q_{tot}^{\bar{\theta}_{i,t-k+1}}(s',\mathbf{a}')}{\sum_{k=1}^{K} c_{k}}).
\end{equation}
with
\begin{equation}
    c_{k} = \max (0, sign(\bar{Q}_{tot}^{\bar{\phi}}(s',\mathbf{a}') - Q_{tot}^{\bar{\phi}_{t-k+1}}(s',\mathbf{a}'))),
\end{equation}
where $c_{k}$ determines whether the global action value should be preserved if it is below average or discarded otherwise. $\bar{Q}_{tot}^{\bar{\phi}}(s',\mathbf{a}')$ is the average of the last $K$ global action values.


\section{Noisy SMAC Settings}
\label{appendix:noisysmac}

In this noisy SMAC environment, we add a random noise for each feature of both the observation and global state, which could be regarded as a kind of interference signal. The random noise is set to be uniformly distributed and its range is $[0, 0.02)$. Although the noise is small, it dramatically raises the overestimation problem for multiagent Q-learning algorithms and seriously impedes the quality of the learned policies, which is not reported in the literature before. The noisy SMAC environment is a good testbed for MARL algorithms to address the multiagent overestimation problem.

\begin{figure}[htbp]
\centering
\subfigure[5m\_vs\_6m]{
\label{5m6mgaussian}
\includegraphics[width=0.239\textwidth]{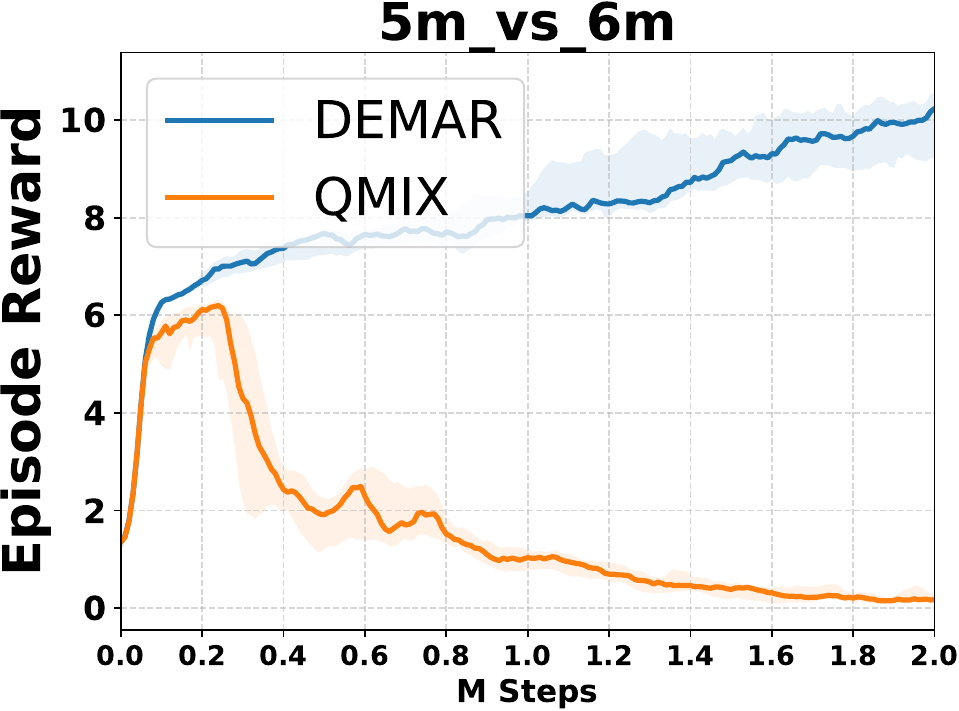}}
\subfigure[$Q_{tot}$ on 5m\_vs\_6m]{
\label{5m6mgaussianoverestimation}
\includegraphics[width=0.239\textwidth]{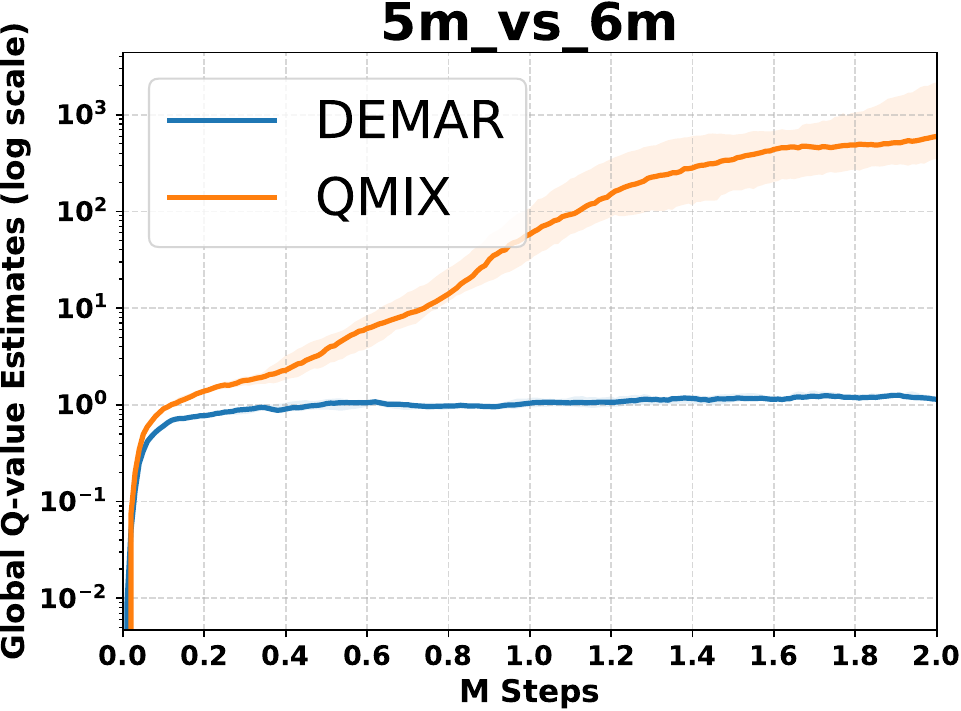}}
\subfigure[10m\_vs\_11m]{
\label{10m11mgaussian}
\includegraphics[width=0.239\textwidth]{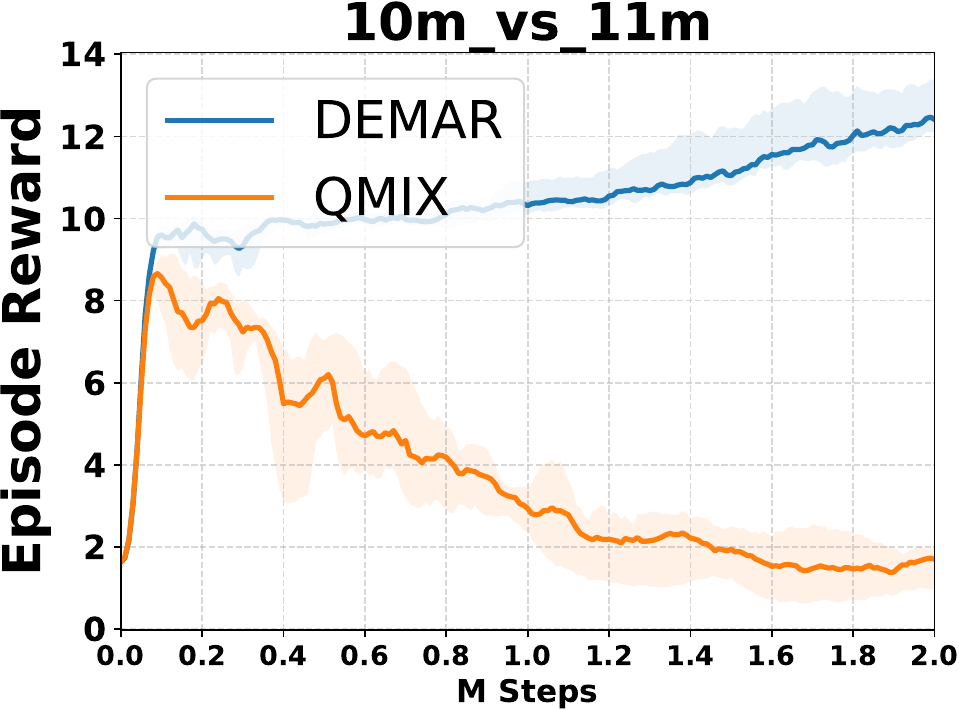}}
\subfigure[$Q_{tot}$ on 10m\_vs\_11m]{
\label{10m11mgaussianoverestimation}
\includegraphics[width=0.239\textwidth]{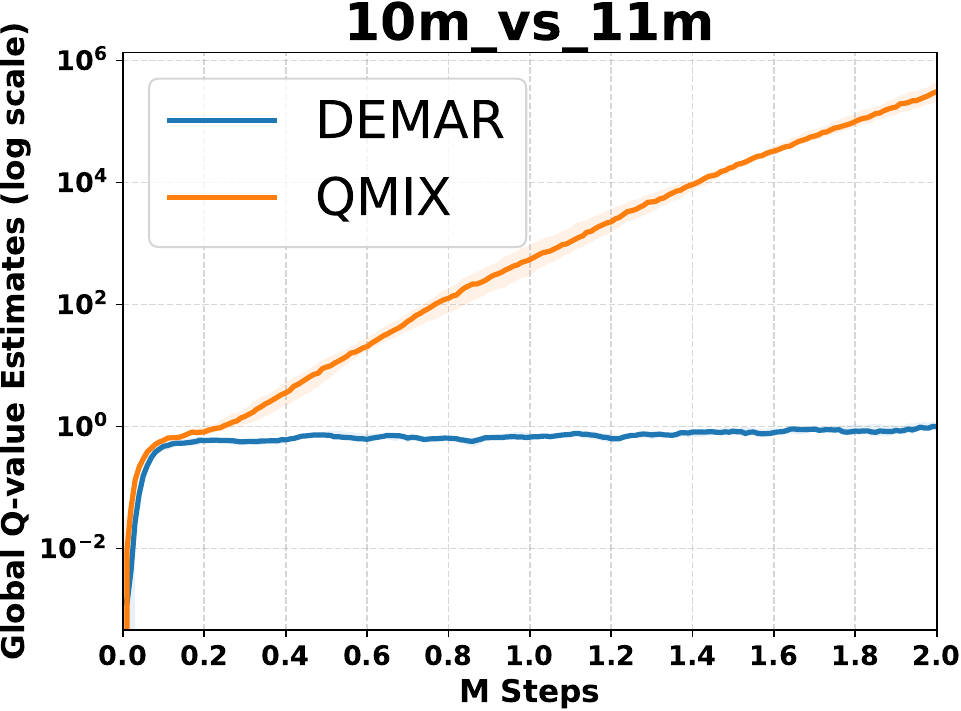}}
\caption{Testing DEMAR on the noisy SMAC environment with the Gaussian distributed noise.}
\label{figure:gaussian}
\end{figure}

Meanwhile, we also test DEMAR's performance on the noisy SMAC with the Gaussian distributed noise. The random noise is set to be Gaussian distributed and its mean and standard deviation are both 0.02. The noise is added to each feature of both the observation and global state. The results are shown in Figure~\ref{figure:gaussian}. As we can see, DEMAR also controls the overestimation well and stabilizes the learning on both the 5m\_vs\_6m and 10m\_vs\_11m with Gaussian noise.

\section{Hyperparameter Settings}
\label{appendix:hypersetting}
As different tasks have different levels of overestimation, we adjust the hyperparameters of each method on each task. To make a fair comparison, we perform the grid search for all baselines around their tuned default values which perform best in their original papers. Specifically, for S-QMIX, we search $\beta \in \{50.0, 5.0, 0.5, 0.05, 0.005\}$. For SM2-QMIX, we search $(\alpha, \omega) \in \{(10.0, 5.0), (10.0, 0.5), (10.0, 0.05), (1.0, 5.0), \\ (1.0, 0.5), (1.0, 0.05), (0.1, 5.0), (0.1, 0.5), (0.1, 0.05)\}$. For WCU-QMIX, we search $w \in \{0.25, 0.5, 0.75\}$. The tuned hyperparameters of each method on each task of MPE as shown in Table~\ref{table:mpehyper}.

\begin{table}[htbp]
\centering
\caption{Hyperparamters of algorithms on MPE.}
\begin{tabular}{|c|c|c|c|}
\hline
DEMAR & simple tag & simple world & simple adversary \\ \hline
$H$     &       3     &     10     &    10       \\ \hline
$N_{\mathbb{H}}$  &      3      &     6         &    4          \\ \hline
$K$     &       1     &      1        &         10     \\ \hline
$N_{\mathbb{K}}$  &      1      &      1        &    4    \\ \hline
$\alpha_{reg}$ &      0.002      &    0.02      &    0.05   \\ \hline \hline
S-QMIX & simple tag & simple world & simple adversary \\ \hline
$\beta$ &      0.05      &      0.5        &       0.005     \\ \hline \hline
SM2-QMIX & simple tag & simple world & simple adversary \\ \hline
$\alpha$ &     0.1       &      10.0        &    0.1   \\ \hline
$\omega$ &     5.0      &       0.05   &     0.5     \\ \hline \hline
WCU-QMIX & simple tag & simple world & simple adversary \\ \hline
$w$ &   0.75   &    0.75    &   0.75  \\ \hline
\end{tabular}
\label{table:mpehyper}
\end{table}

We also use the grid search on each SMAC tasks for all baselines. For SM2-QMIX, we search $(\alpha, \omega) \in \{(10.0, 5.0), (10.0, 0.5), (10.0, 0.05), \\ (1.0, 5.0), (1.0, 0.5), (1.0, 0.05), (0.1, 5.0), (0.1, 0.5), (0.1, 0.05)\}$. For WCU-QMIX, we search $w \in \{0.25, 0.5, 0.75\}$. For Sub-Avg-QMIX, we search $K \in \{3,5,10\}$. The tuned hyperparameters of each method on each task of SMAC as shown in Table~\ref{table:smachyper}.

\begin{table}[htbp]
\centering
\caption{Hyperparameters of algorithms on SMAC.}
\begin{tabular}{|c|c|c|c|c|}
\hline
DEMAR    &  5m\_vs\_6m &   2s3z     & 3s5z   & 10m\_vs\_11m \\ \hline
$H$      &      3      &     3      &    10   &     4        \\ \hline
$N_{\mathbb{H}}$  &      2      &     2      &    9   &     3        \\ \hline
$K$      &      1      &     1      &    1   &     1        \\ \hline
$N_{\mathbb{K}}$  &      1      &     1      &    1   &     1        \\ \hline
$\alpha_{reg}$ &      0.002  &     0.002  &  0.001 &     0.01     \\ \hline \hline
SM2-QMIX & 5m\_vs\_6m &   2s3z     & 3s5z   & 10m\_vs\_11m \\ \hline
$\alpha$ &     1.0       &      10.0        &    10.0  &  1.0  \\ \hline
$\omega$ &     0.5      &       0.05   &     5.0   &   5.0   \\ \hline \hline
WCU-QMIX & 5m\_vs\_6m &   2s3z     & 3s5z   & 10m\_vs\_11m \\ \hline
$w$ &   0.75   &    0.75    &   0.75    &     0.75     \\ \hline \hline
Sub-Avg-QMIX & 5m\_vs\_6m &   2s3z     & 3s5z   & 10m\_vs\_11m \\ \hline
$K$ &      10      &      3        &       3     &     3    \\ \hline 
\end{tabular}
\label{table:smachyper}
\end{table}

As there are five hyperparameters for DEMAR to tune, the grid search for DEMAR needs massive computation budgets. Instead of grid search, we use a heuristic sequential searching to search the hyperparameters for DEMAR, which greatly reduces the load of tuning. The sequential searching is as follows. First, we adjust the $\alpha_{reg}$ to see whether the overestimation is controlled to avoid extremely large global Q-values. Next, when overestimation avoids being extremely large, we adjust the $H$ and $N_{\mathbb{H}}$ to further limit the overestimation. Finally, if the overestimation still exists to influence the learning, we adjust the $K$ and $N_{\mathbb{K}}$. We use this hyperparameter sequential searching in both environments for DEMAR. In most cases, the first two steps could successfully mitigate the overestimation problem. On the other hand, DEMAR could return to vanilla QMIX by setting hyperparameters as $H=N_{\mathbb{H}}=K=N_{\mathbb{K}}=1$ and $\alpha_{reg}=0$ if the vanilla algorithm does not have the overestimation issue in the environment, which demonstrates the flexibility of our method.

\section{Ablation Study of Dual Ensembled Q-learning and Hypernet Regularizer}
\label{sec:ablationdual}

Here we also conduct the ablation study with different Q-network ensemble sizes and subset sizes for target $Q_{tot}$ and $Q_{i}$ in the dual ensembled Q-learning. We use the \textit{simple\_adversary} as the tested scenario. The standard hyperparameter setting of DEMAR on \textit{simple\_adversary} are $H=10$, $N_{\mathbb{H}}=4$, $K=10$, and $N_{\mathbb{K}}=4$ for dual ensembled Q-learning while $\alpha_{reg}=0.05$ for hypernet regularizer. We test $H \in \{5, 10, 15\}$, $N_{\mathbb{H}} \in \{1, 2, 4, 6, 8\}$, $K \in \{5, 10, 15\}$, and $N_{\mathbb{K}} \in \{1, 2, 4, 6, 8\}$ separately while keeping $\alpha_{reg}=0.05$ unchanged. Results are shown in Figure~\ref{returnhablation}-\ref{returnnkablation} respectively. As we can see, most hyperparameter settings also mitigate severe overestimation and the standard hyperparameter setting performs best. Meanwhile, we also include more results of ablation studies of K in the 5m\_vs\_6m and 2s3z. We try different combinations of $K$ and $N_{\mathbb{K}}$ on both 5m\_vs\_6m and 2s3z in SMAC. We keep other hyperparameters fixed. The results are provided in Figure~\ref{figure:ab5m6mK}.

\begin{figure}[htbp]
\centering
\subfigure[Episode return of $H$]{
\label{returnhablation}
\includegraphics[width=0.238\textwidth]{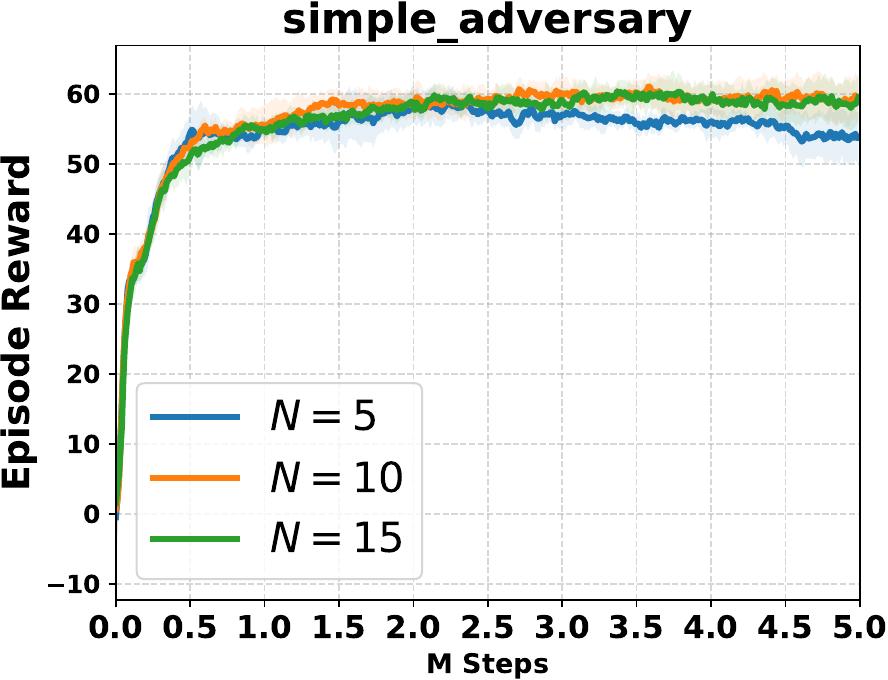}}
\subfigure[Episode return of $N_{\mathbb{H}}$]{
\label{returnnhablation}
\includegraphics[width=0.238\textwidth]{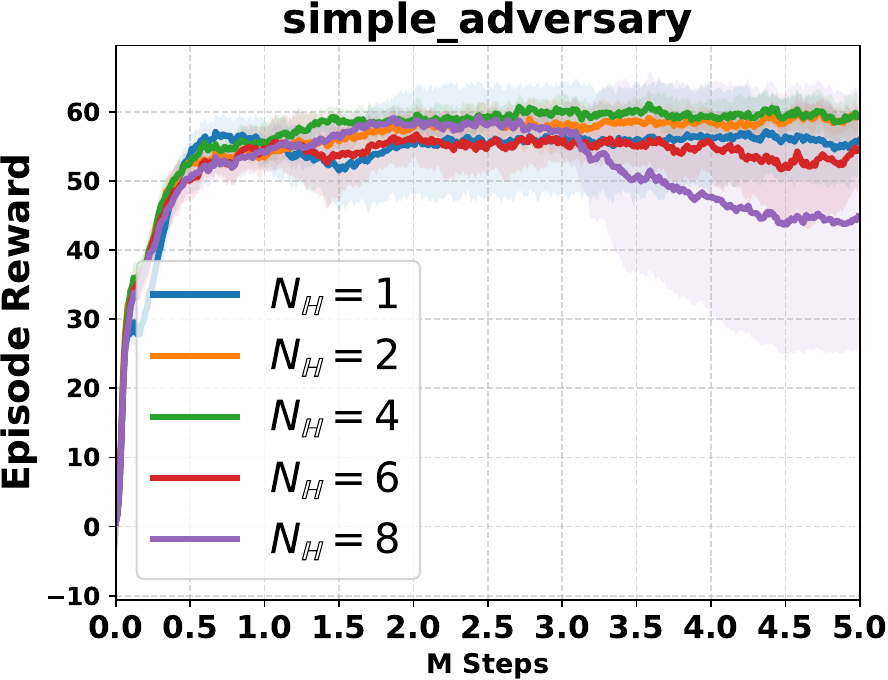}}
\subfigure[Episode return of $K$]{
\label{returnkablation}
\includegraphics[width=0.238\textwidth]{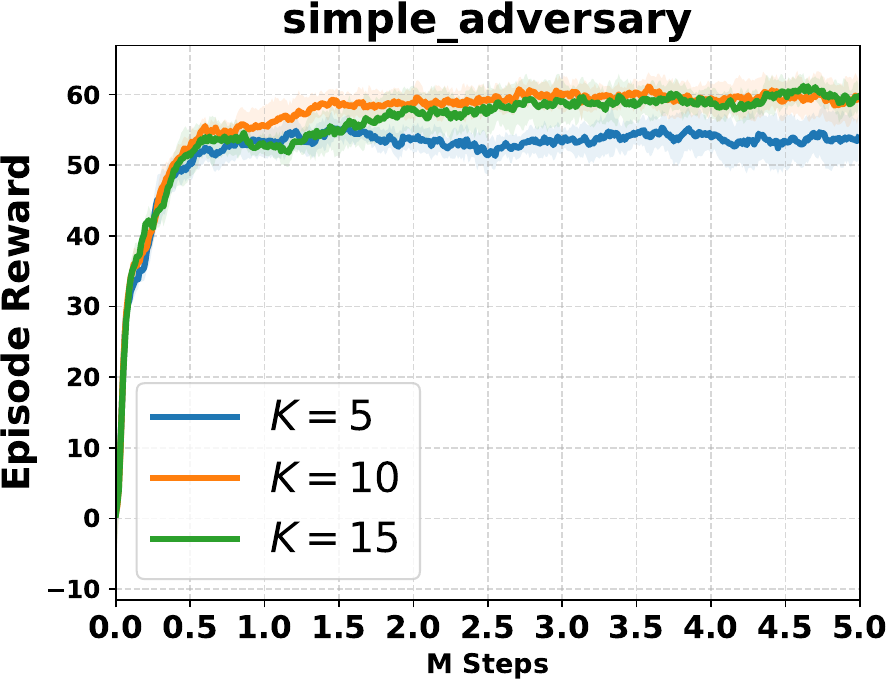}}
\subfigure[Episode return of $N_{\mathbb{K}}$]{
\label{returnnkablation}
\includegraphics[width=0.238\textwidth]{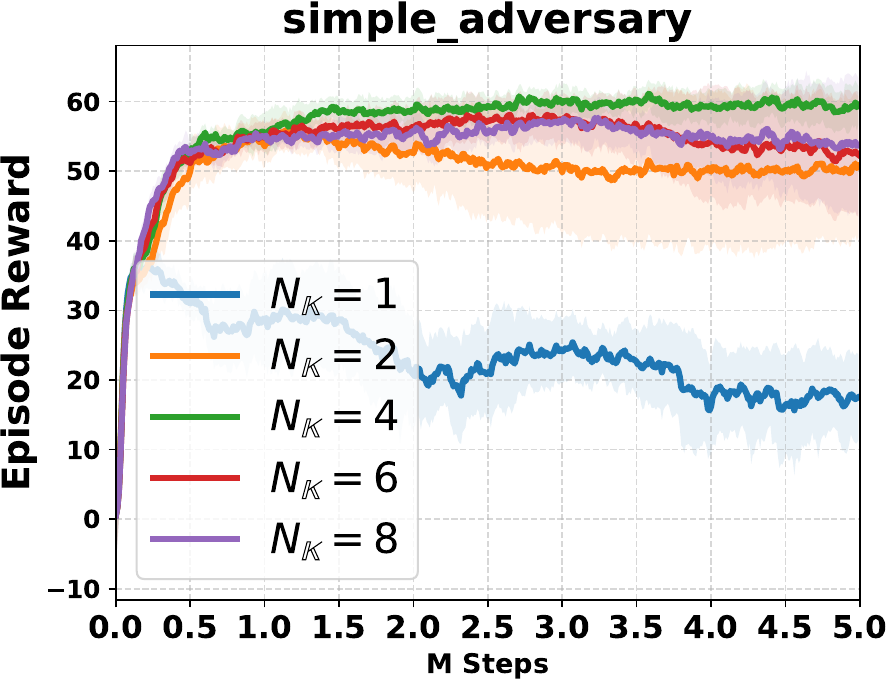}}
\subfigure[$Q_{tot}$ of $H$]{
\label{qtothablation}
\includegraphics[width=0.238\textwidth]{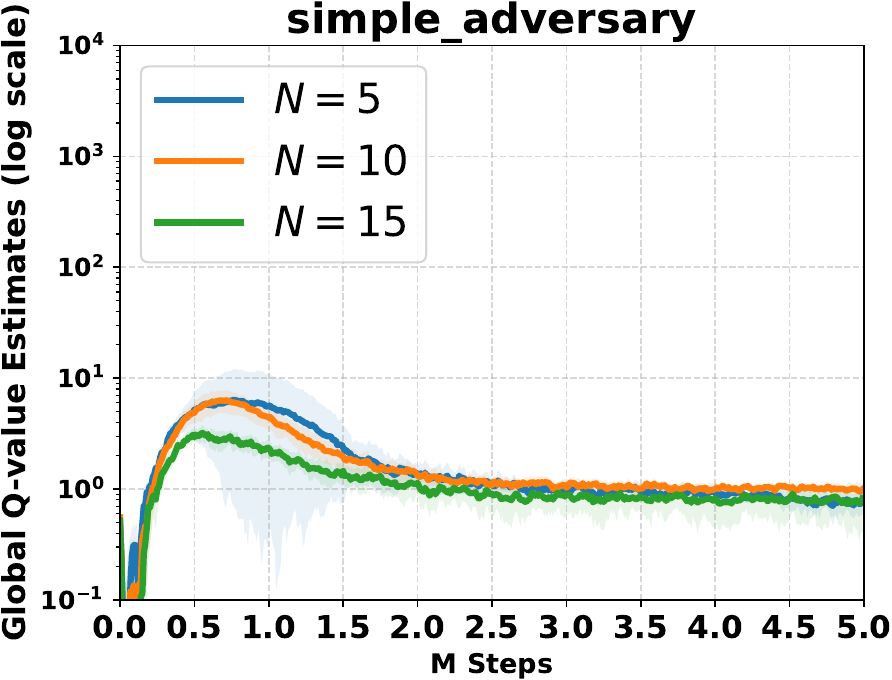}}
\subfigure[$Q_{tot}$ of $N_{\mathbb{H}}$]{
\label{qtotnhablation}
\includegraphics[width=0.238\textwidth]{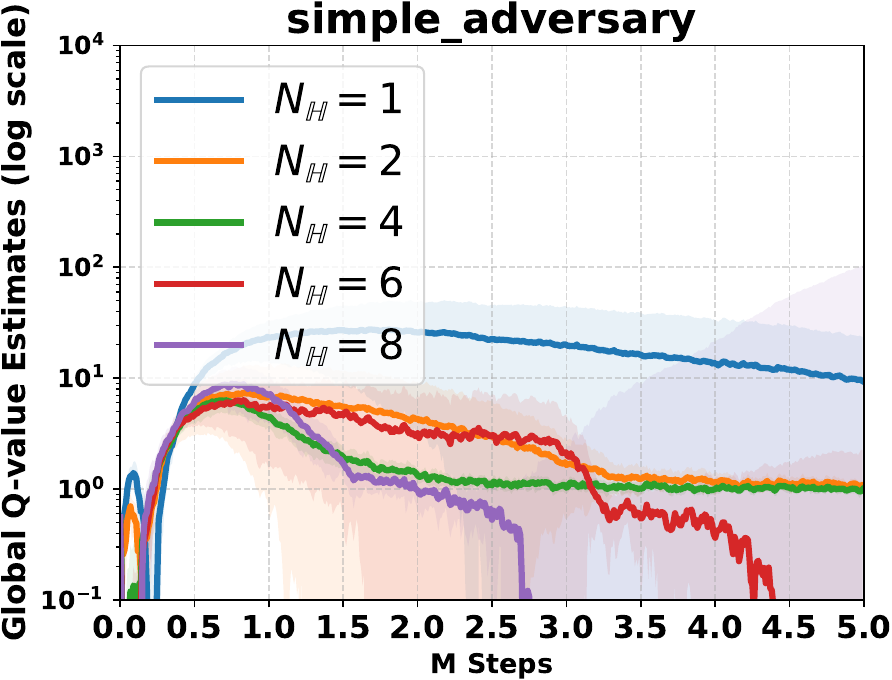}}
\subfigure[$Q_{tot}$ of $K$]{
\label{qtotkablation}
\includegraphics[width=0.238\textwidth]{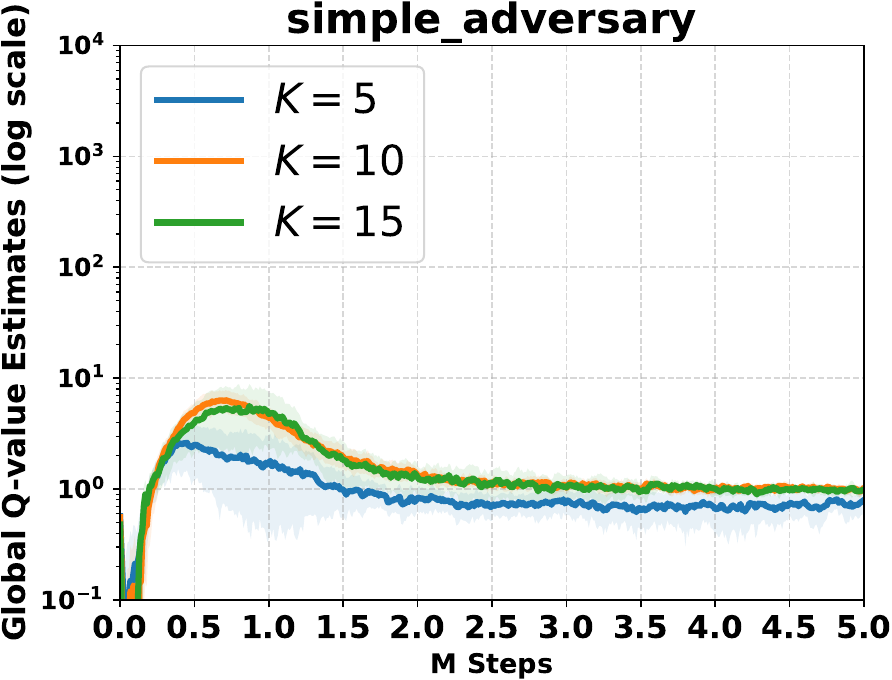}}
\subfigure[$Q_{tot}$ of $N_{\mathbb{K}}$]{
\label{qtotnkablation}
\includegraphics[width=0.238\textwidth]{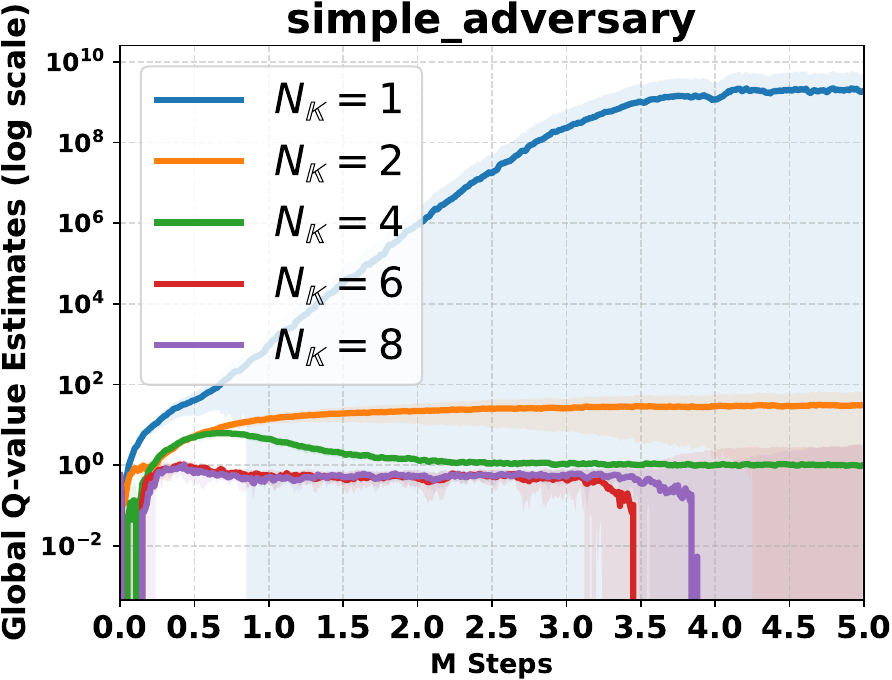}}
\caption{Ablation of dual ensembled Q-learning on the \textit{simple\_adversary} task.}
\label{figure:dualablation}
\end{figure}

\begin{figure}[htbp]
\centering
\subfigure[Episode return on 5m\_vs\_6m]{
\includegraphics[width=0.242\textwidth]{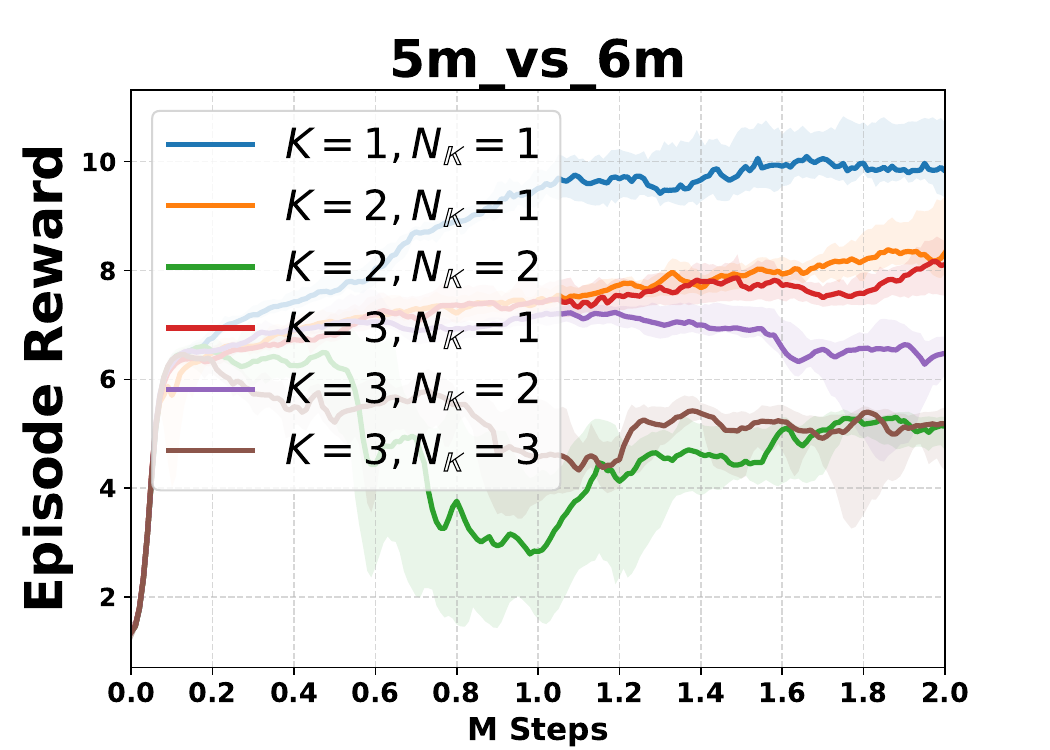}}
\subfigure[Test win rate on 5m\_vs\_6m]{
\includegraphics[width=0.242\textwidth]{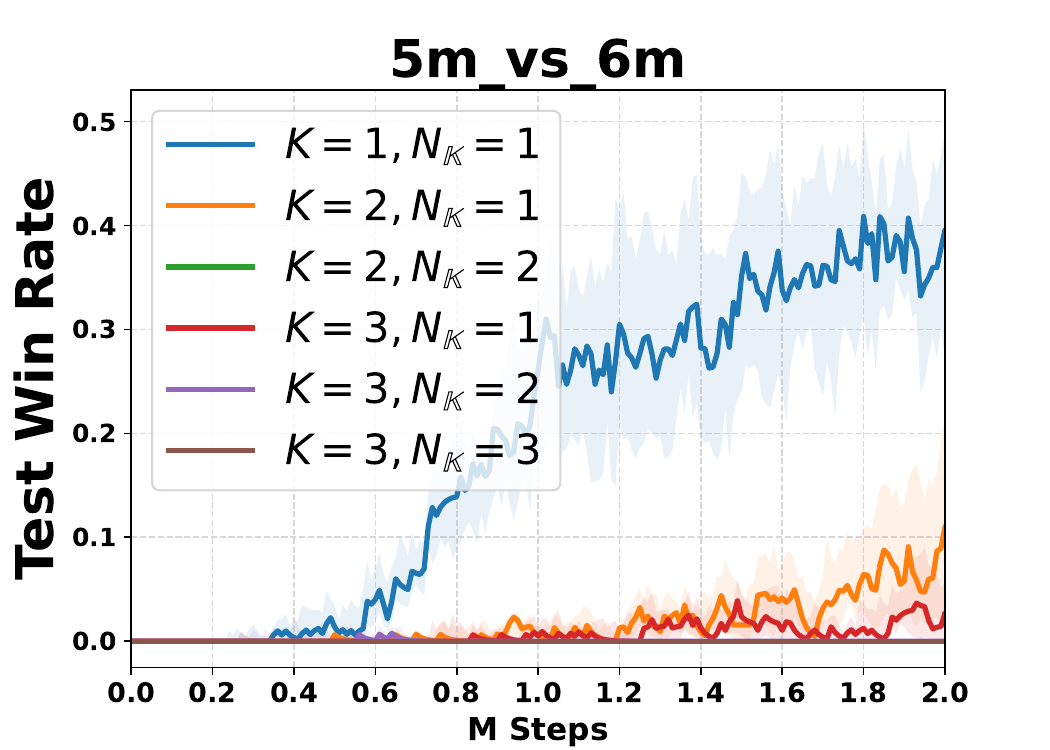}}
\subfigure[Episode return on 2s3z]{
\includegraphics[width=0.242\textwidth]{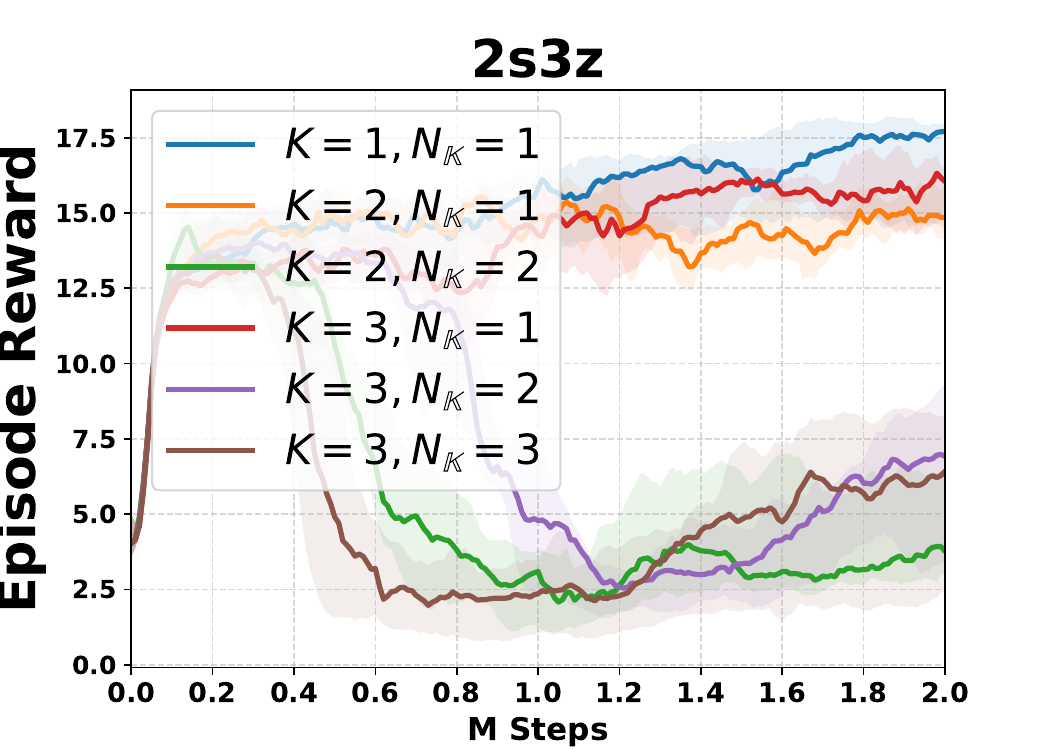}}
\subfigure[Test win rate on 2s3z]{
\includegraphics[width=0.242\textwidth]{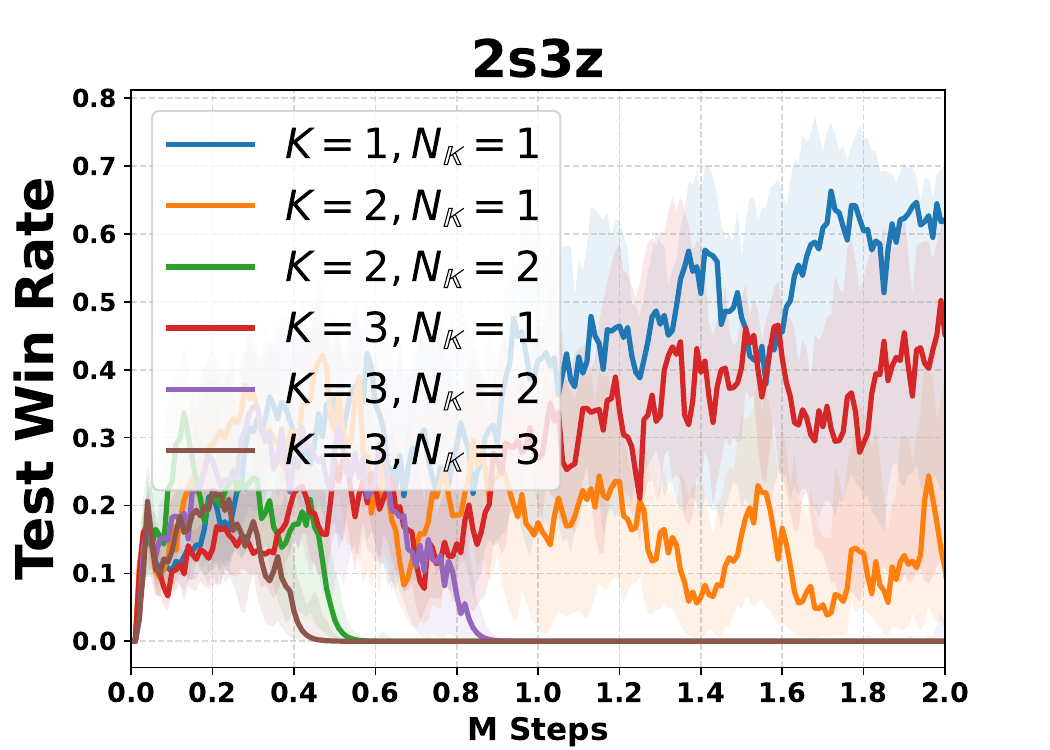}}
\caption{Ablation of different combinations of $K$ and $N_{\mathbb{K}}$ on the \textit{5m\_vs\_6m} and \textit{2s3z} tasks.}
\label{figure:ab5m6mK}
\end{figure}

\begin{figure}[htbp]
\centering
\subfigure[Episode return of $\alpha_{reg}$]{
\label{returnhyperablation}
\includegraphics[width=0.3\textwidth]{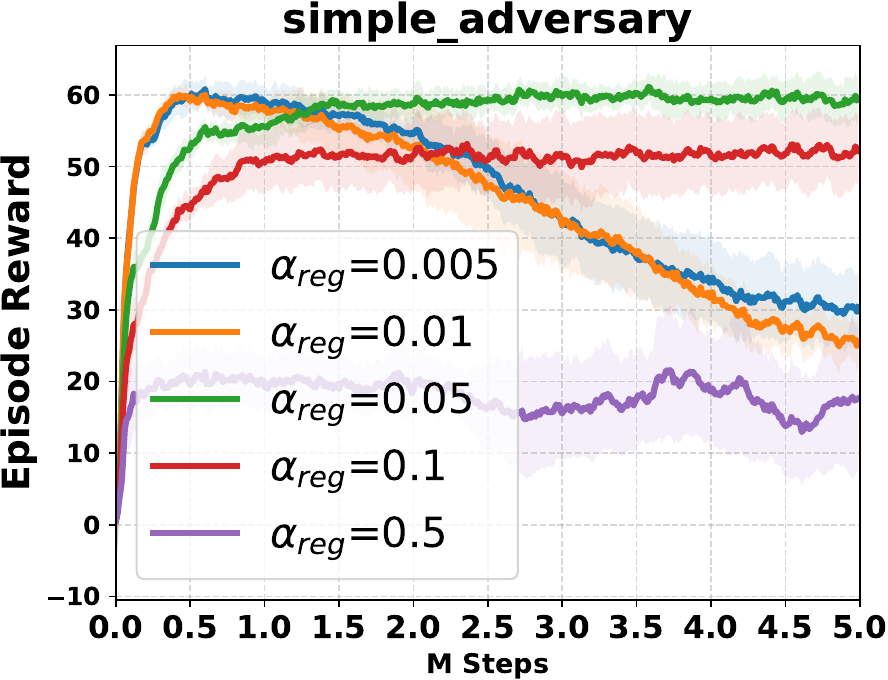}}
\subfigure[$Q_{tot}$ of $\alpha_{reg}$]{
\label{qtothyperablation}
\includegraphics[width=0.3\textwidth]{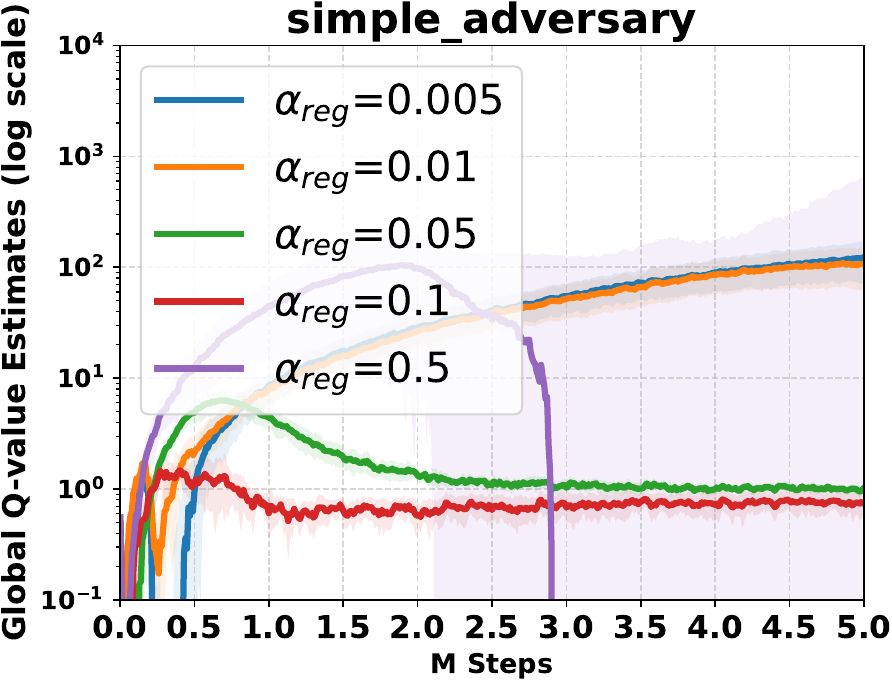}}
\caption{Ablation of hypernet regularizer on the \textit{simple\_adversary} task.}
\label{figure:hyperablation}
\end{figure}

Besides the ablation on dual ensembled Q-learning, we also conduct the ablation study on the coefficient of hypernet regularizer $\alpha_{reg}$. The standard hyperparameter setting of $\alpha_{reg}$ on \textit{simple\_adversary} is 0.05. We test $\alpha_{reg} \in \{0.005, 0.01, 0.05, 0.1, 0.5\}$ with the standard hyperparameter setting of dual ensembled Q-learning. The ablation results are shown in Figure~\ref{figure:hyperablation}. We could see that, if $\alpha_{reg}$ is too small, the severe overestimation cannot be fully tackled. If the $\alpha_{reg}$ is too large, although the severe overestimation is tackled, the algorithm performance would be affected. This indicates $\alpha_{reg}$ need to be tuned carefully.

\section{Test Win Rate in SMAC}

As the test win rate is the most popular performance metric in SMAC, we also include it as an additional performance metric for better evaluation in SMAC. The results are shown in Figure~\ref{figure:smac_winrate}.

\begin{figure}[htbp]
\centering
\subfigure[Test win rate on 5m\_vs\_6m]{
\includegraphics[width=0.242\textwidth]{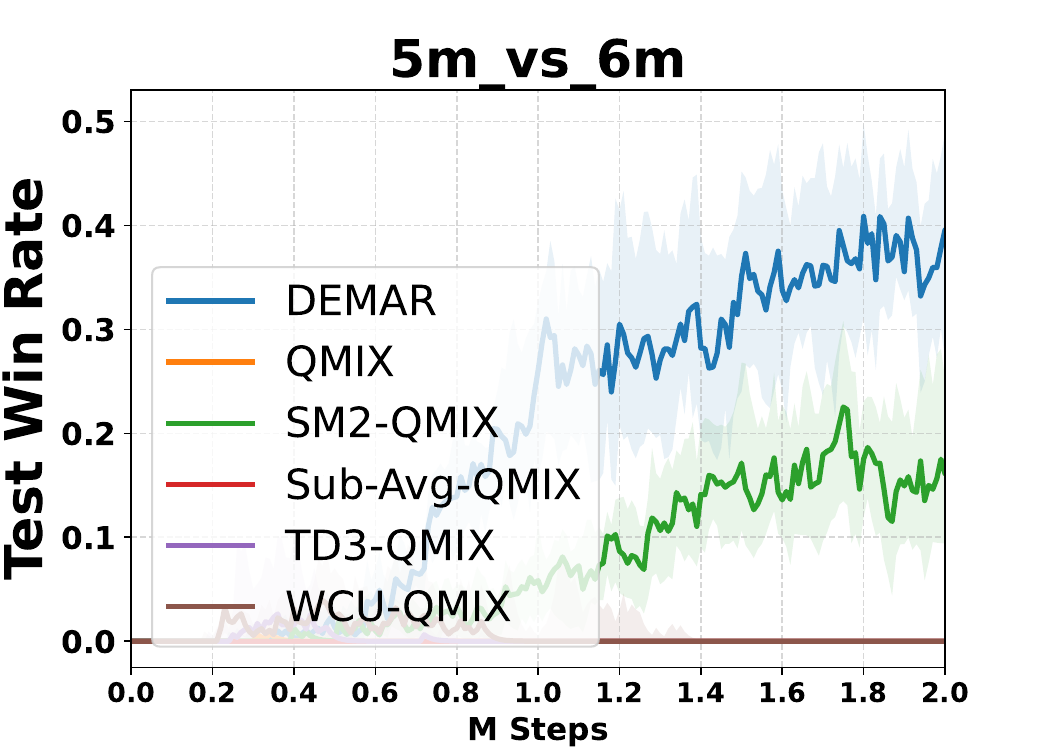}}
\subfigure[Test win rate on 10m\_vs\_11m]{
\includegraphics[width=0.242\textwidth]{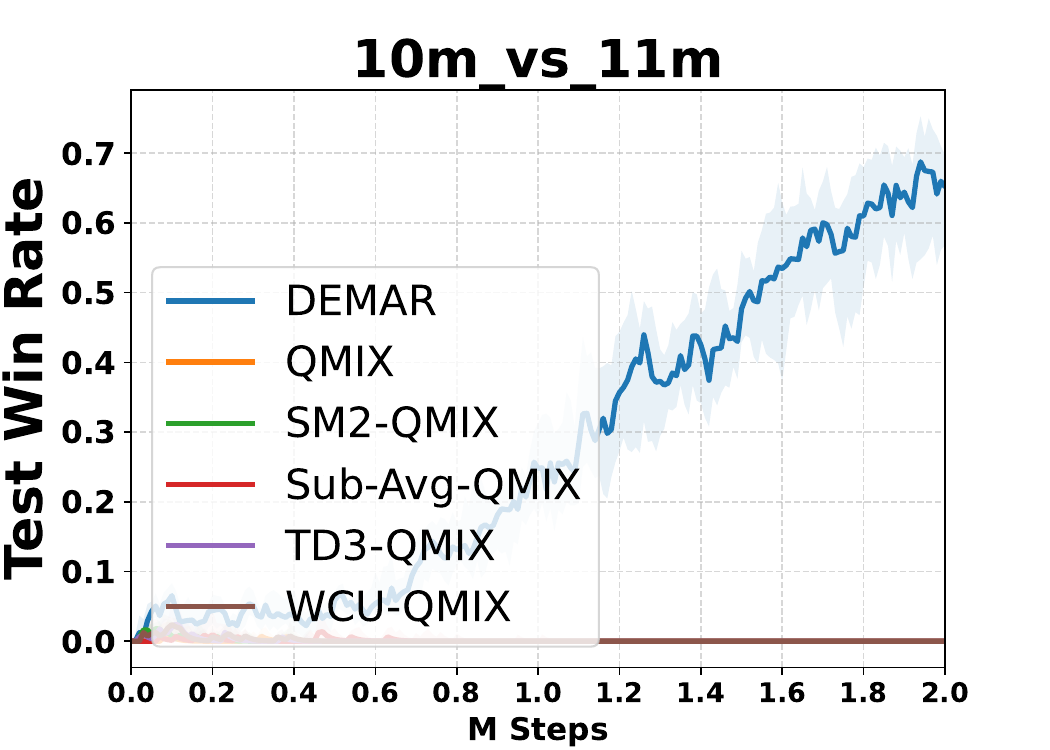}}
\subfigure[Test win rate on 2s3z]{
\includegraphics[width=0.242\textwidth]{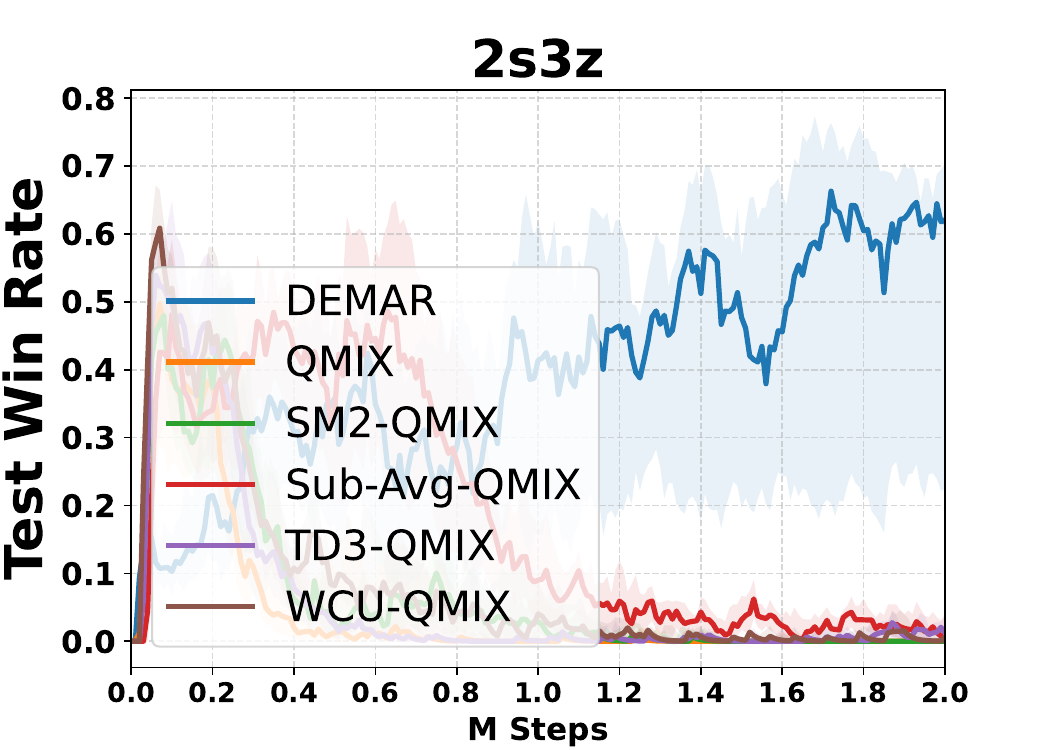}}
\subfigure[Test win rate on 3s5z]{
\includegraphics[width=0.242\textwidth]{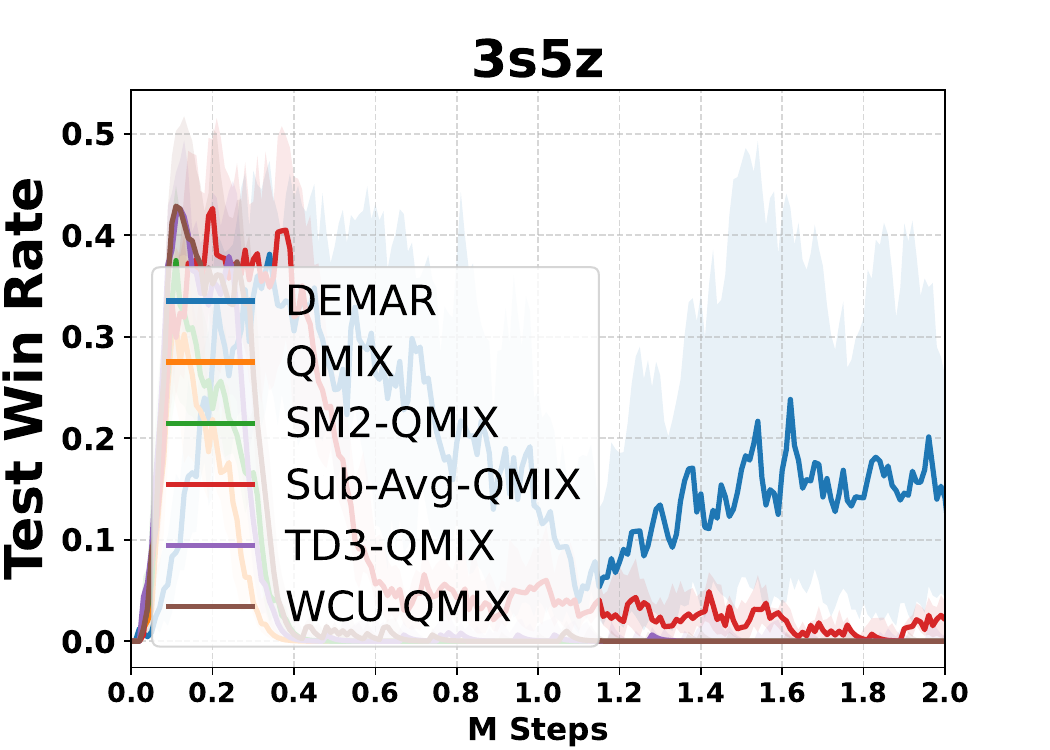}}
\caption{The test win rate on different maps in SMAC.}
\label{figure:smac_winrate}
\end{figure}

\section{Broader Impacts}
\label{appendix:impact}
MARL is a powerful paradigm that can model real-world systems, such as autonomous driving, network optimization, and energy distribution. The proposed DEMAR could stabilize the learning of existing multiagent value-mixing Q-learning algorithms, thus increasing the practicability of MARL in real-world applications especially when robust learning is required. However, when applied to real-world tasks, the learning process of MARL with DEMAR still needs some explorations which may lead to unsafe situations. On the other hand, there still exists the risk of using MARL with DEMAR to do unethical actions such as using MARL to perform network attacks.

\section{Limitations}
\label{appendix:limit}
Our study may have limitations under extensive consideration. First, our method is not suitable for the policy-based MARL algorithms. The full adaption of DEMAR to policy-based MARL methods may need further efforts and we list this as one of our future works. Second, there are five hyperparameters for DEMAR, which need more hyperparameter tuning although we have developed a heuristic sequential searching to help the tuning. Third, as the ensemble-based methods use multiple networks, DEMAR has a larger network parameter size although it uses the same network architecture compared with other non-ensemble baselines. This is the inherent property of ensemble-based methods such as TD3 \cite{fujimoto_addressing_2018} and REDQ \cite{chen_randomized_2021}, and so does DEMAR. It may limit the scalability of DEMAR for a large number of agents.

\end{document}